\documentclass[11pt,a4paper]{article}

\usepackage[margin=1in]{geometry}
\usepackage[english]{babel}
\usepackage{amsmath,amssymb,amsthm}
\usepackage[hidelinks]{hyperref}
\usepackage{cleveref}
\usepackage{cite}
\usepackage{xcolor}
\usepackage[linesnumbered,ruled,vlined,noend]{algorithm2e}
\usepackage{tikz}
\usepackage{placeins}
\usetikzlibrary{positioning,fit,backgrounds,calc}

\newtheorem{theorem}{Theorem}
\newtheorem{lemma}[theorem]{Lemma}
\newtheorem{proposition}[theorem]{Proposition}
\newtheorem{claim}[theorem]{Claim}
\newtheorem{corollary}[theorem]{Corollary}

\theoremstyle{definition}
\newtheorem{definition}[theorem]{Definition}

\theoremstyle{remark}
\newtheorem{remark}[theorem]{Remark}

\newenvironment{claimproof}
{\begin{proof}}
	{\end{proof}}

\crefname{theorem}{theorem}{theorems}
\Crefname{theorem}{Theorem}{Theorems}
\crefname{lemma}{lemma}{lemmas}
\Crefname{lemma}{Lemma}{Lemmas}
\crefname{proposition}{proposition}{propositions}
\Crefname{proposition}{Proposition}{Propositions}
\crefname{claim}{claim}{claims}
\Crefname{claim}{Claim}{Claims}
\crefname{corollary}{corollary}{corollaries}
\Crefname{corollary}{Corollary}{Corollaries}
\crefname{definition}{definition}{definitions}
\Crefname{definition}{Definition}{Definitions}
\crefname{remark}{remark}{remarks}
\Crefname{remark}{Remark}{Remarks}
\crefname{algocf}{Algorithm}{Algorithms}
\Crefname{algocf}{Algorithm}{Algorithms}

\newcommand{\dist}{\operatorname{dist}}

\allowdisplaybreaks

\title{Testing Induced-Subgraph Freeness in Outerplanar Graphs \\under the Random-Neighbor Oracle}

\author{
  Pan Peng\footnote{
	School of Computer Science and Technology, University of Science and Technology of China.  Email: \href{ppeng@ustc.edu.cn}{ppeng@ustc.edu.cn}}\\
  \and
  Kefan Yu\footnote{
	School of Computer Science and Technology, University of Science and Technology of China.  Email: \href{ykf97@mail.ustc.edu.cn}{ykf97@mail.ustc.edu.cn}}\\
}

\date{}

\begin{document}
	\hypersetup{pageanchor=false}
	\maketitle
	
	\begin{abstract}
		We prove that, for every fixed nonempty graph \(H\), induced-\(H\)-freeness
		is testable with \(\varepsilon^{-O_H(1)}\) queries on outerplanar graphs with
		no maximum-degree bound in the \emph{random-neighbor model}, where each query
		at a vertex returns a uniformly random neighbor. Thus, the query complexity
		is polynomial in \(1/\varepsilon\) and independent of the number \(n\) of
		vertices. Previously, the best bound known for this problem was the
		\(\operatorname{poly}(\log n)\)-query guarantee that follows from the general
		outerplanar-graph tester of Babu, Khoury, and Newman~\cite{BKN16} in the
		stronger \emph{adjacency-list model}, which provides exact degree queries and
		indexed access to neighbors.
		
		Our tester has \emph{two-sided error}, which
		is necessary in general: induced-\(P_3\)-freeness has no one-sided
		constant-query tester in the random-neighbor model, even on outerplanar graphs of
		maximum degree two.
	\end{abstract}
	
	\thispagestyle{empty}

\newpage
\hypersetup{pageanchor=true}
\pagenumbering{arabic}

	\section{Introduction}

	\emph{Property testing} distinguishes objects satisfying a property
	from objects far from satisfying it, while examining only a small
	portion of the input; no guarantee is required for intermediate objects.
	Since the foundational work of Blum, Luby, and Rubinfeld~\cite{BLR90}
	and Rubinfeld and Sudan~\cite{RS96}, and the introduction of graph property
	testing by Goldreich, Goldwasser, and Ron~\cite{GGR98}, property testing has developed into a central area of
	sublinear-time algorithms.  Graphs provide one of its richest settings,
	because the appropriate notion of local access depends strongly on the
	density and structure of the input.
	
	The study of graph property testing began in the \emph{dense graph model}, where
	the graph is accessed through adjacency queries. Alon, Fischer, Newman, and
	Shapira~\cite{AFNS06} characterized the testable graph properties. Alon and
	Shapira~\cite{AS08} gave a corresponding characterization for natural
	\emph{one-sided testers}, which never reject a
	graph that has the property, and, in particular, proved the testability of
	every \emph{hereditary property}, namely, every property preserved under
	vertex deletion. These results illustrate the central role of global
	regularity methods in the dense model.
	
	Goldreich and Ron~\cite{GR97} initiated a different line of research for
	sparse graphs of bounded maximum degree.  Here local neighborhoods, rather
	than dense regularity, become the basic source of information.  A sequence
	of results established constant-query testability for increasingly broad
	sparse graph families and properties.  Czumaj, Shapira, and
	Sohler~\cite{CSS09} studied hereditary properties on nonexpanding
	bounded-degree graphs; Benjamini, Schramm, and Shapira~\cite{BSS10} proved
	that every \emph{minor-closed property}, which is preserved under deleting
	vertices or edges and contracting edges, is testable; and Newman and
	Sohler~\cite{NS11} showed that
	every property of a \emph{hyperfinite graph family} is testable. Such a
	family can be split into bounded-size components by deleting a small
	fraction of its edges. The role of hyperfiniteness in bounded-degree testing was
	further clarified by Fichtenberger, Peng, and Sohler~\cite{FPS19}.
	\emph{Partition oracles}, which give local access to such a decomposition,
	were introduced in this context by Hassidim et al.~\cite{HKNO09}. They have
	become a fundamental algorithmic tool, and the
	polynomial-query partition oracle of Kumar, Seshadhri, and
	Stolman~\cite{KSS21} gives particularly efficient local access to
	decompositions of minor-free graphs. For bounded-degree planar graphs, Basu, Kumar, and
	Seshadhri~\cite{BKS21} further investigated the complexity of testing
	arbitrary properties and the central role of graph isomorphism.
	
	Much less is understood once the maximum-degree restriction is removed.
	Even in a sparse graph class, one vertex may have degree close to \(n\).
	Consequently, a tester cannot inspect every bounded-radius neighborhood with a
	constant number of queries.
	
	One line of work gives constant-query one-sided testers for specific
	properties, including bipartiteness and many forms of
	\emph{non-induced subgraph freeness}, where the forbidden copy need not be
	induced. Czumaj,
	Monemizadeh, Onak, and Sohler~\cite{CMOS11} gave a constant-query
	tester for bipartiteness in planar graphs in the \emph{random-neighbor
		model}, where a query at a vertex returns a uniformly random neighbor.
	Czumaj and
	Sohler~\cite{CS19} then showed that, for every fixed graph \(H\), non-induced
	\(H\)-freeness is one-sided testable on planar graphs and, more generally,
	on minor-free graphs. Esperet and Norin~\cite{EN22} later showed that every
	\emph{monotone property}, which is preserved under edge deletion, is
	testable on every proper minor-closed class. Levi and
	Shoshan~\cite{LS21} also gave a polynomial-query tester for the nonmonotone
	property of Hamiltonicity in minor-free graphs.
	
	Recent work extends non-induced \(H\)-freeness testing to wider sparse
	classes. Awofeso, Greaves, Lachish, Levi, and
	Reidl~\cite{awofeso2025sufficient} gave a general sufficient
	condition for one-sided random-neighbor testing. Humeau, Kant{\'e}, Mock, Picavet, and 
	Vigny~\cite{humeau2025testing} gave constant-query testers for classes of
	bounded expansion. Lachish, Levi, Newman, and
	Reidl~\cite{LLNR26} gave a characterization for bounded-degeneracy graphs.
	It identifies the properties that admit one-sided random-neighbor testers.
	For a fixed
	\(2\)-connected graph \(H\), their result says that non-induced
	\(H\)-freeness is testable exactly when \(H-S\) is connected for every
	independent set \(S\subseteq V(H)\). There are also lower bounds: Eden,
	Levi, and Ron~\cite{ELR24} proved polynomial lower bounds in \(n\) for
	testing several fixed cycle-freeness properties in the
	\emph{adjacency-list model}, where the tester can query exact degrees and
	indexed neighbors. These results show that both the graph class and the
	property matter.
	
	A second line of work gives sublinear testers for \emph{all} properties in a certain class of sparse graphs. Kusumoto and Yoshida~\cite{KY14} showed that every
	property of forests is testable with \(\operatorname{poly}(\log n)\)
	queries. Babu, Khoury, and Newman~\cite{BKN16} proved the same type of bound
	for every property of outerplanar graphs in the adjacency-list model.
	
	These two lines of work do not settle the complexity of testing a basic
	hereditary property in sparse graphs: \emph{induced-subgraph freeness}. For a
	fixed graph \(H\), this property requires that no vertex set induce a copy of
	\(H\). It has been studied extensively in the dense graph
	model~\cite{AFKS99,AS06induced,AS08}. The general outerplanar-graph tester
	above applies to this property, but uses \(\operatorname{poly}(\log n)\)
	queries and the stronger adjacency-list model. On the other hand, the known
	constant-query results for fixed forbidden subgraphs concern non-induced
	copies.

	Induced copies introduce a basic difficulty that is absent in the
	non-induced setting. A non-induced copy can be certified by finding its
	edges. An induced copy can be certified only after checking both its edges
	and its nonedges. The latter is difficult under random-neighbor access:
	sampling can reveal an edge, but a bounded number of samples cannot certify
	that two high-degree vertices are nonadjacent. Thus the central question is
	whether the fixed size of the forbidden pattern, together with the structure
	of the input class, is enough to overcome the lack of direct access to
	nonedges.

	Outerplanar graphs are a natural first class for this question. They allow
	unbounded degrees, and hence the difficulty above remains, but they have a strong
	local decomposition. Moreover, every outerplanar property has a
	polylogarithmic-query tester, while many non-induced subgraph freeness
	properties have constant-query testers. We show that induced-\(H\)-freeness
	can be tested with a number of random-neighbor queries independent of \(n\).
	
	\subsection{Our Contributions}
	
	Our main result gives a random-neighbor tester for induced-\(H\)-freeness
	for every fixed forbidden graph \(H\): despite the lack of direct access to
	nonedges, the tester's query complexity is independent of \(n\).
	It has \emph{two-sided error}: it may err with bounded probability both
	when the input has the property and when it is far from the property.
	In the stronger adjacency-list model, we also obtain a one-sided tester
	with the same asymptotic query bound.
	Distance is measured by the number of edge additions and deletions on the
	same vertex set, normalized by \(n\). For a proximity parameter
	\(\varepsilon\in(0,1)\), an input is \(\varepsilon\)-far if more than
	\(\varepsilon n\) such edits are needed to obtain the property.
	We use this distance as in earlier work on bipartiteness
	and non-induced subgraph freeness in sparse graphs~\cite{CMOS11,CS19,humeau2025testing}. Only the input graph is promised to
	be outerplanar; the comparison graph after editing need not be outerplanar. This separation between the input promise and the comparison class also
	appears in the minor-free Hamiltonicity tester of Levi and
	Shoshan~\cite{LS21}, with a different distance normalization. 
	Both main theorems use this convention. In particular, we do not claim a
	tester for distance to \emph{outerplanar} induced-\(H\)-free graphs.
	
	\begin{theorem}[Random-neighbor model]
		\label{thm:main}
		Let \(H\) be any fixed nonempty graph and let \(\varepsilon\in(0,1)\)
		be the proximity parameter. Then induced-\(H\)-freeness is two-sided
		testable on outerplanar graphs in the random-neighbor model with query
		complexity \(\varepsilon^{-O_H(1)}\).
	\end{theorem}
	
	The exponent implicit in \(O_H(1)\) depends only on \(H\), and is independent of \(n\) and \(\varepsilon\). Our current analysis gives an exponent much larger than \(100\), and we have made no attempt to optimize it. \Cref{thm:main} improves the general
	$\operatorname{poly}(\log n)$-query guarantee of Babu, Khoury, and
	Newman~\cite{BKN16}, and it works in the weaker random-neighbor model, which
	provides neither exact degrees nor indexed access to adjacency lists.
	More specifically, it extends constant-query forbidden-subgraph testing from
	non-induced copies to induced copies. The general theorem of Esperet and
	Norin~\cite{EN22} applies to monotone properties, and therefore it does not cover this
	setting: deleting an edge can create a forbidden induced copy.
	
	The two-sided error is unavoidable in general for this family of properties.
	Already for the
	three-vertex path \(P_3\), one-sided testing is impossible with constant query
	complexity under random-neighbor access.
	
	\begin{proposition}[A one-sided random-neighbor barrier]
		\label{prop:one-sided-barrier}
		Induced-\(P_3\)-freeness is not testable with one-sided error and constant query complexity in the random-neighbor
		model on outerplanar graphs of maximum degree at most two.
	\end{proposition}
	
	We defer the proof of \Cref{prop:one-sided-barrier} to
	\Cref{app:one-sided-barrier}. The main algorithmic difficulty at high-degree vertices is that nonedges cannot be checked directly. This lower bound reflects a second limitation of random-neighbor access:
even at degree two, finitely many samples cannot certify that the
complete neighborhood has been seen. Thus, one-sided error is impossible in general even when
high degrees are absent.
	
	In the stronger adjacency-list model, exact degree and indexed-neighbor
	queries remove this obstacle. The structural algorithm underlying
	\Cref{thm:main} then yields a one-sided tester.
	
	\begin{theorem}[Adjacency-list model]
		\label{thm:sparse-main}
		Let \(H\) be any fixed nonempty graph and let \(\varepsilon\in(0,1)\)
		be the proximity parameter. Then induced-\(H\)-freeness is one-sided
		testable on outerplanar graphs in the adjacency-list model with query
		complexity \(\varepsilon^{-O_H(1)}\).
	\end{theorem}

	\subsection{Proof overview and main ideas}
	
	We first give a one-sided tester in the adjacency-list model. We then
	simulate this tester in the random-neighbor model. The proof has three main ingredients.
	First, we modify an outerplanar decomposition so that each resulting
	component keeps all of its original internal edges. Second, we give a
	procedure that finds a connected induced copy by exploring low-degree
	vertices and sampling branches around a high-degree vertex. Third, for
	disconnected \(H\), we show how to choose copies of its components that are
	disjoint and have no edges between them. Finally, we use random-neighbor
	samples to recover all bounded-degree lists needed by these procedures.
	
	\paragraph*{The decomposition.}
	We make use of a decomposition of the outerplanar input graph \(G\). Fix a
	degree threshold \(d=\varepsilon^{-O(1)}\). The goal is to remove only
	\(O(\varepsilon n)\) edges and obtain a spanning subgraph
	\(G'\subseteq G\) such that each connected component of \(G'\) contains at
	most one vertex whose degree in \(G\) is greater than \(d\). Babu, Khoury,
	and Newman~\cite{BKN16} showed that a decomposition of this form exists.
	They also gave a local algorithm for the decomposition. It computes
	the local pieces belonging to the connected component of \(G'\) containing
	a queried vertex, without constructing the whole decomposition or listing
	the possibly large component.
	
	We follow the local-decomposition framework of~\cite{BKN16}, with two changes.
	First, on the bounded-degree subgraph of \(G\) induced by the low-degree
	vertices, we use the more recent partition oracle of Kumar,
	Seshadhri, and Stolman~\cite{KSS21}. A partition oracle computes the part
	containing a queried vertex in a partition into small connected parts,
	consistently across queries that use the same fixed seed. This oracle uses
	\(\operatorname{poly}(d/\xi)\) queries when its parameter is \(\xi\);
	with constant probability, at most \(\xi dn\) edges join different parts. We use
	\(\xi=\Theta(\varepsilon/d)\), and hence the resulting bound is
	\(\operatorname{poly}(d^2/\varepsilon)\).
	
	Second, we make a change to the construction of~\cite{BKN16} that is
	important for induced copies. After deleting edges between low-degree
	parts and separating the high-degree vertices attached to each part, we
	restore every original edge whose endpoints lie in the same resulting component. Thus,
	\(G'[C]=G[C]\) for every component \(C\) of \(G'\). Consequently, a vertex set contained
	in one component induces exactly the same graph in \(G'\) and in \(G\).
	This lets us search inside \(G'\) without creating a false induced
	copy by deleting an internal edge.

We remark that the restriction to outerplanar inputs is important for this step. In particular, the same decomposition does not extend to planar graphs. Consider the planar graph \(K_{2,N}\). Its two vertices on the size-two
	side have degree \(N\) and are joined by \(N\) edge-disjoint paths of
	length two. For any fixed degree threshold \(d\), taking \(N>d\) makes
	both vertices high-degree. Putting them in different connected components
	requires at least \(N\) edge deletions. For any fixed
	\(\varepsilon<1\) and sufficiently large \(N\), this exceeds
	\(\varepsilon(N+2)\). Thus planar graphs need not admit a decomposition
	with this small edit budget and at most one original high-degree vertex
	per component. This is an obstruction to our decomposition, not an
	impossibility result for induced-freeness testing on planar graphs.
	A planar extension would need a different way to handle components with
	several high-degree vertices and the nonedges between them.
	
	\paragraph*{Connected forbidden graphs.}
	Suppose first that \(H\) is connected. We give a tester for induced-$H$-freeness with two steps. The
	first step uses random sampling and bounded exploration to produce a small
	set of \emph{candidates}. A candidate consists of \(|V(H)|\) distinct
	explored vertices together with a proposed correspondence to \(V(H)\). The second
	step performs an \emph{exact check}: it determines every edge and nonedge
	among the candidate vertices and keeps the candidate exactly when these
	vertices induce \(H\) in the original graph \(G\).
	
	The exact check is possible because of the decomposition of the outerplanar graph. The tester
	discards every candidate that contains two high-degree vertices. Thus, every
	pair of vertices in a remaining candidate has a low-degree endpoint. The
	complete adjacency list of this endpoint tells us whether the pair is an
	edge or a nonedge. Consequently, the tester never reports a false induced copy. In
	particular, if \(G\) is induced-\(H\)-free, every candidate fails the exact
	check and the tester always accepts.
	
	It remains to explain why the first step finds a candidate when \(G\) is far
	from induced-\(H\)-free. On the event that the decomposition removes few
	edges, \(G'\) is still far from the property. A maximal family of edge-disjoint induced copies
	then gives linearly many copies of \(H\) in \(G'\). We remove some copies
	while keeping linearly many of them. In the remaining family, every vertex
	used by a copy has a positive fraction of its incident edges in the
	family; this fraction depends only on \(H\) and \(\varepsilon\). This property ensures that sampling a
	neighbor of a participating high-degree vertex reveals an edge of one of
	the copies with a probability that does not depend on \(n\).
	
	Every connected component of \(G'\) contains at most one high-degree
	vertex, and hence each packed copy either has only low-degree vertices or has one
	high-degree root. At least one of these two cases contains many copies. In
	the first case, a uniformly random vertex hits the low-degree vertices of
	the packed copies with constant probability. A bounded breadth-first
	exploration from this vertex then reveals the whole copy.
	
	In the second case, we again start from a uniformly sampled low-degree
	vertex in one of the packed copies. A bounded exploration reaches its unique
	high-degree root \(a\). Removing the corresponding root from \(H\) splits
	\(H\) into a bounded number of connected branches. We sample neighbors of
	\(a\) and explore the low-degree branches reached by these samples. The
	samples need not come from the same packed copy. This is the main issue in
	the high-degree case. We prove that each sampled branch conflicts with only
	a bounded number of branches of another type. Thus, independently
	sampled branches are pairwise compatible with constant probability, and
	together with \(a\) they form a candidate copy of \(H\). Repeating the
	search a constant number of times and applying the exact check gives the
	connected tester.
	
	\paragraph*{Disconnected forbidden graphs.}
	If \(H\) is disconnected, finding its components one at a time is not
	enough. The copies found must be disjoint, and there must be no edges
	between them. The main question is whether we can keep finding new
	copies of each component after excluding the copies already chosen.
	To express the possible obstruction, we define \emph{small removal witnesses}:
	for a connected graph \(F\), a small removal witness is a bounded family
	of vertex-disjoint induced \(F\)-copies such that, after their vertices
	are removed, a specified number of edge additions or deletions makes the
	remaining graph induced-\(F\)-free. The family size depends only on \(H\),
	and the further edit budget is a suitably chosen multiple of \(\varepsilon n\).
	Our structural lemma, \Cref{far}, gives two cases.
	
	In the first case, no component \(H_i\) of \(H\) has a small removal witness.
	Thus, after the vertices of any bounded family of induced \(H_i\)-copies
	are removed, the remaining graph is still far from induced-\(H_i\)-freeness.
	We can then run the connected finder repeatedly and obtain \(K+1\)
	vertex-disjoint candidates for each component \(H_i\), where \(K\)
	depends only on \(H\) and is defined in \Cref{sec:disconnected-structure}.
	We form
	an auxiliary colored graph whose vertices are these candidates. The color
	records the associated component \(H_i\) of \(H\); repeated isomorphic
	components of \(H\) have separate indices and colors. Two candidates are
	adjacent when an edge of \(G\) runs between them. Copies are disjoint
	within each family, but different families may intersect. A greedy
	selection first makes the families globally vertex-disjoint, as shown in
	\Cref{sec:disconnected-tester}. The auxiliary graph is then a minor of \(G\),
	obtained by contracting the connected candidates and deleting other
	vertices and edges, and is outerplanar because this class is minor-closed. We then choose one candidate of each color so that no
	two chosen candidates are adjacent. Such a choice is called a
	\emph{rainbow independent transversal}: it meets each color class exactly
	once. \Cref{upperbound} guarantees that it exists when each color has
	enough candidates. Their union is
	then an induced copy of \(H\).
	
	In the second case, exactly one component has a small removal witness.
	This component has at least two vertices, and its isomorphism type occurs
	only once in \(H\). \Cref{far} shows that one copy from a small removal
	witness has a unique high-degree vertex, which we call its root. Associated
	with this copy is a separated low-degree region \(B\subseteq V(G)\),
	disjoint from the copy and with no edge to it in the original graph \(G\).
	Inside \(G'[B]\), the graph is still far from induced-\(H_i\)-freeness for
	each remaining nontrivial component \(H_i\). If \(H\) has isolated vertices,
	\(B\) also contains linearly many possible singleton candidates.
	This second case is essential even for forests: a large star together with
	many isolated vertices is far from induced-\(P_3\cup K_1\)-free, although
	all its induced \(P_3\)'s share the center. We give this example and its
	exact edit distance at the beginning of \Cref{sec:disconnected-structure}.
	
	Membership in \(B\) can be tested from bounded-degree adjacency lists. We
	use a low-degree version of the connected finder to find the nontrivial
	components inside \(B\), and uniform vertex sampling finds singleton
	candidates for the isolated components of \(H\); these vertices need not
	be isolated in \(G\). The same selection theorem, \Cref{upperbound}, gives
	one candidate for each remaining
	component, with no intersections or edges between the chosen candidates. Since \(B\) has no edge to the witness copy, the selected
	candidates together with that witness form an induced copy of \(H\).
	
	These two cases give a one-sided adjacency-list tester. A
	\emph{threshold-list query} returns the full list of a vertex of degree
	at most \(d\), and reports high degree otherwise. We use it as shorthand
	for the following adjacency-list operations. In the
	adjacency-list model, one degree query and at most \(d\) indexed-neighbor
	queries answer it exactly. A uniform neighbor of a high-degree root is also
	obtained exactly by choosing a uniform index in its adjacency list. The analysis of these two cases establishes \Cref{thm:sparse-main}
	before we turn to the weaker model.
	
	\paragraph*{Random-neighbor simulation.}
	The random-neighbor oracle does not provide exact degrees or lists. We
	simulate a threshold-list query by taking repeated random-neighbor samples.
	Seeing \(d+1\) distinct neighbors proves that the degree is high. Otherwise,
	enough samples recover the full list of a low-degree vertex with high
	probability. This is the standard coupon-collector bound for seeing every
	neighbor. We store the first recovered answer for each vertex and reuse it
	on later calls. On the event that all these answers are correct, the local
	partition queries therefore refer to the same graph.
	
	The tester is \emph{adaptive}: its next query can depend on earlier
	answers. We therefore need a guarantee for the whole execution, not just
	for a single list. We show that, with high probability, every recovered
	threshold-list answer used in the whole execution is simultaneously
	correct. We call this the recovery event. On this event, the
	simulated execution has the same sequence of queries and answers, and hence
	the same output, as the adjacency-list execution. Direct samples at
	high-degree roots already use
	the random-neighbor oracle and need no simulation. The recovery event can
	fail with small probability, which changes the tester from one-sided to
	two-sided error and gives \Cref{thm:main}.
    
	\paragraph*{Organization.}
	\Cref{pre} fixes the notation, defines the two access models, and introduces
	the exact threshold-list shorthand used by the adjacency-list tester.
	\Cref{sec:global-partition} gives the induced-preserving outerplanar
	decomposition and shows how to access it locally. \Cref{sec:induced-finder}
	describes and analyzes the tester for connected forbidden graphs.
	\Cref{sec:disconnected-structure} then proves the structural dichotomy needed
	for disconnected forbidden graphs. \Cref{sec:disconnected-tester} turns this
	dichotomy into the disconnected-case tester and completes the proof of
	\Cref{thm:sparse-main}. \Cref{sec:rn-simulation} recovers bounded-degree
	neighborhoods from random-neighbor samples, uses them to simulate the
	adjacency-list tester, and proves \Cref{thm:main}.
	\Cref{app:one-sided-barrier} proves the one-sided lower bound, and
	\Cref{app:decomposition} gives the deferred proof of the decomposition
	lemma. \Cref{app:ai-disclosure} records our use of AI assistance from GPT-5.6 sol.

	\section{Preliminaries and access models}\label{pre}
	This section sets the notation and query conventions used in the proofs. Its
	main tool is the bounded-degree partition-oracle theorem stated at the end of
	the section. The theorem gives a local view of a partition of the graph's vertex set
	into small connected parts. We first define the two access models, the
	threshold-list shorthand, and the graph-edit distance, and then state
	the partition-oracle theorem.
	
	All graphs are finite, undirected, and simple. For a graph $X$, we write
	$V(X)$ and $E(X)$ for its vertex and edge sets, and write $\Gamma_X(v)$ and
	$\deg_X(v)$ for the neighborhood and degree of a vertex $v$. We use $X-S$
	for the subgraph induced by $V(X)\setminus S$. For a set $S\subseteq V(X)$,
	let $\Gamma_X(S):=\bigcup_{v\in S}\Gamma_X(v)$. For disjoint vertex sets
	$A,B\subseteq V(X)$, let $E_X(A,B)$ be the edges with one endpoint in each set.
	Throughout the paper, $G$ denotes the outerplanar \emph{input graph} of the
	tester, and $n:=|V(G)|$. Write $[n]:=\{1,\ldots,n\}$ and identify $V(G)$ with $[n]$.  A graph is
	induced-$H$-free if no vertex set induces a graph isomorphic to $H$.
	
	Fix the nonempty forbidden graph $H$. We write $H=m_0K_1\cup m_1F_1\cup\cdots\cup m_sF_s$, where $m_0\ge0$, $m_r\ge1$ for $r\in[s]$, and
	$F_1,\ldots,F_s$ are pairwise non-isomorphic connected graphs with at
	least two vertices. We call such connected graphs \emph{nontrivial}.
	Here $K_1$ is the graph consisting of a single vertex. Thus, $m_0$ is
	the multiplicity of $K_1$, while $m_r$ is the
	multiplicity of the component type $F_r$. We also list the individual connected components as
	$H_1,\ldots,H_k$, where $k=m_0+\sum_{r=1}^s m_r$.
	Unlike the types \(F_r\), the graphs \(H_i\) need not be pairwise
	non-isomorphic. For example, if \(H=2P_3\cup K_1\), then two distinct
	indices in this list represent copies of \(P_3\). This distinction matters
	when the tester builds candidate families for each indexed component of \(H\).
	
\paragraph{Query access models}
	A tester is a randomized algorithm that receives $n$ and the proximity
	parameter $\varepsilon$, for the fixed property being tested, and accesses
	$G$ through an oracle: a black box answering the permitted graph queries.
	An access model specifies which queries are available. A random-neighbor
	query $\mathsf{RN}(v)$ returns an element of $\Gamma_G(v)$ uniformly at
	random, or $\perp$ if $v$ is isolated; distinct calls use independent randomness.
	A degree query returns $\deg_G(v)$, and an indexed-neighbor query returns
	the $i$th neighbor in a fixed ordering, or $\perp$ if that index is out of range.
	The \emph{random-neighbor model} permits random-neighbor queries, whereas
	the \emph{adjacency-list model} permits degree and indexed-neighbor queries.
	In both models $n$ is known, and the tester may choose any identifier in
	$[n]$, including a uniformly random one, without a graph query.
	The adjacency-list model simulates a random-neighbor query exactly by
	first learning the degree and then choosing a uniform neighbor index.
	
	For an integer threshold \(t\ge0\), we use the following shorthand for
	exact adjacency-list access:
	\[
	\mathsf{TL}_t(v)
	:=
	\begin{cases}
		(\mathsf{Low},\Gamma_G(v)^\uparrow), & \deg_G(v)\le t,\\
		\mathsf{High}, & \deg_G(v)>t,
	\end{cases}
	\]
	where \(\Gamma_G(v)^\uparrow\) is the complete neighborhood of \(v\), listed
	in increasing vertex-ID order. We call this a \emph{threshold-list query}.
	It is simulated in the adjacency-list model by first querying \(\deg_G(v)\): if the degree is at
	most \(t\), the tester queries all indexed neighbors, and otherwise it
	returns \(\mathsf{High}\). Therefore, one threshold-list query costs at most
	\(t+1\) adjacency-list queries.
	
	For the algorithms in \Cref{sec:induced-finder,sec:disconnected-tester},
	\emph{threshold-list access} with threshold $t$ means that we can simulate
	$\mathsf{TL}_t$ queries, sample uniform vertices, and sample uniform neighbors
	of vertices classified as high-degree. These operations are simulated without
	error in the adjacency-list model. In \Cref{sec:rn-simulation}, we simulate
	the threshold-list queries used by these algorithms with a constant number
	of random-neighbor queries per operation and a small probability of error;
	``constant'' may depend on the fixed thresholds and desired error bounds.
	Direct random-neighbor samples require no simulation.
	
    \paragraph{Property testing}
	For graphs $X$ and $Y$ on the same vertex set, denote their \emph{edit distance} by
	$\dist(X,Y):=|E(X)\triangle E(Y)|$. This is the number of edge additions and
	deletions needed to turn $X$ into $Y$. For a property $\mathcal P$, let
	$\dist(X,\mathcal P)$ be the minimum of $\dist(X,Y)$ over graphs $Y\in\mathcal P$
	on that vertex set; if there is no such graph, take the distance to be $+\infty$.
	Let $\varepsilon\in(0,1)$ be the \emph{proximity parameter}. We call $X$
	\emph{$\varepsilon$-far} from $\mathcal P$ if
	$\dist(X,\mathcal P)>\varepsilon|V(X)|$.
	An $\varepsilon$-tester for $\mathcal P$ must accept every promised input
	in $\mathcal P$ and reject every promised $\varepsilon$-far input, each with
	probability at least $2/3$; it has no required behavior on intermediate inputs.
	A tester allowing error in both cases is \emph{two-sided}.
	A \emph{one-sided tester} accepts every input in $\mathcal P$ with probability
	one and rejects every $\varepsilon$-far input with probability at least $2/3$.
	
	The promise that the input is outerplanar applies only to \(G\). Unless
	stated otherwise, the comparison graph \(X\in\mathcal P\) is not required to
	be outerplanar. In particular, an input is \(\varepsilon\)-far from
	induced-\(H\)-freeness when it is far from every induced-\(H\)-free graph on the
	same vertex set. This distance definition is used when the structural proofs construct
	a nearby graph by both adding and deleting edges.

	When we consider a graph $X$ on a subset of $V(G)$, we still measure
	distances using the original $n=|V(G)|$. In component searches we use $F$
	for the connected pattern and $\beta$ for its possibly smaller proximity
	parameter: saying that $X$ is $\beta$-far from induced-$F$-freeness means
	that more than $\beta n$ edge changes are needed, even when $X$ has fewer
	than $n$ vertices. The residual graphs are actual induced subgraphs;
	we retain the original normalization rather than changing their vertex sets.
	By \emph{constant} we mean a quantity independent of the input graph $G$,
	but possibly dependent on $H$ and $\varepsilon$. In particular,
	$\varepsilon^{-O_H(1)}$ denotes a bound polynomial in $1/\varepsilon$,
	with exponent depending only on $H$. The query complexity of a tester
	is the number of graph-oracle queries it makes, as a function of $H$,
	$\varepsilon$, and the input size; computation on already obtained data is not counted.
	
	A \emph{partition oracle} is a deterministic local procedure with
	adjacency-list access to a graph $X$ of maximum degree at most $d$, a
	parameter $\xi\in(0,1)$, random access to a seed $\mathbf R$, and a queried
	vertex $v$. For fixed $X,d,\xi,\mathbf R$, it returns the part containing
	$v$ in one partition $\mathcal P_{\mathbf R}$ of $V(X)$ into sets inducing
	connected subgraphs of $X$, without constructing the whole partition.
	Thus it is a procedure using the graph oracle, rather than an additional
	primitive graph query. The shared seed is sampled once and then fixed;
	all probabilities below are over this initial choice. We use the following
	result as a black box; see~\cite[Definition 1.1 and Theorem 1.2]{KSS21full}.
	
	\begin{theorem}[Kumar--Seshadhri--Stolman~\cite{KSS21}]\label{partition}
		Let $\mathcal C$ be a fixed proper minor-closed graph class
		(that is, it excludes some fixed graph as a minor). There is one partition-oracle
		procedure, uniform over $d\ge1$, parameters $\xi\in(0,1)$, and graphs
		$X\in\mathcal C$ of maximum degree at most $d$, with these guarantees:
		\begin{enumerate}
			\item For every fixed seed, its answers refer to a single partition
			$\mathcal P_{\mathbf R}$ of $V(X)$ into sets inducing connected subgraphs,
			each of size $\operatorname{poly}(d/\xi)$.
			\item A query at $v$ returns its entire part using
			$\operatorname{poly}(d/\xi)$ adjacency-list queries to $X$.
			\item With probability at least $2/3$ over $\mathbf R$, at most
			$\xi d|V(X)|$ edges join different parts of this single partition.
		\end{enumerate}
		The polynomial bounds may depend on $\mathcal C$.
	\end{theorem}
	
	Let $\mathbf{findPartition}(v,\mathbf R)$ denote this local procedure
	with $X,d,\xi$ suppressed. Its seed $\mathbf R$ is fixed throughout one
	execution. Let $q_{\mathrm{part}}(\xi,d)=\operatorname{poly}(d/\xi)$ bound
	both the number of queries it makes to $X$ and its output size.
	Consistency, part size, and query bounds hold for every seed; only the
	global bound on crossing edges is probabilistic. The stated success
	probability also follows from the expected-cut guarantee in~\cite{KSS21full}
	by using a constant-factor smaller parameter and Markov's inequality.

	\section{Decomposing the graph into induced components}\label{sec:global-partition}
	This section gives the two decomposition tools used by all later algorithms.
	The main structural result is \Cref{lem:structural}: after deleting only a
	small number of edges, every component has at most one high-degree vertex and
	induces the same graph in the original graph and in the decomposed graph. The
	main algorithmic result is \Cref{simulate}, which gives exact local access to
	this decomposed graph with a number of queries that does not depend on \(n\).
	
	The idea is to partition the subgraph induced by vertices of degree at most
	$d$ in $G$ into connected pieces with at most $s$ vertices,
	separate the high-degree vertices attached to each piece, and then restore all
	original edges inside every resulting component. The restoration step is what
	preserves induced subgraphs. For local access, each vertex is assigned a short
	component label, and an edge is retained exactly when its endpoints have
	the same label. We first define the global construction and state its
	guarantees. We then give the local procedure and prove its guarantee. The
	proof of \Cref{lem:structural} is deferred to
	\Cref{app:decomposition}.
	
	For a graph $X$ and a vertex set $S\subseteq V(X)$, an \emph{$S$-multiway cut} is an edge set whose deletion leaves no two vertices of $S$ in the same component.  Whenever a minimum multiway cut is not unique, we choose the lexicographically first minimum cut under the global ordering of unordered pairs of vertex identifiers.  This convention is used by both the global construction and its local simulation.

	Babu, Khoury, and Newman show that every outerplanar graph is close to
	a union of pieces of bounded vertex count, each component containing at most
	one high-degree root~\cite[Algorithm 1 and Theorem 3.5]{BKN16}.
	We modify their construction in two ways: replace their bounded-degree
	partition oracle by~\cite{KSS21}, and restore every original edge internal
	to a component after the cutting phase. The replacement changes the part-size
	bound to $s=q_{\mathrm{part}}(\varepsilon/(4d),d)$ and the cut probability;
	the counting argument uses this new $s$. Restoration cannot increase the
	deletion budget or merge components, because only already-internal edges
	are restored. A component may contain many bounded pieces attached to
	its root, so the whole component need not have bounded vertex count.
	We verify these claims in \Cref{app:decomposition} rather
	than invoking Theorem 3.5 unchanged. After restoration, the bounded pieces need not have the same form as
	in their decomposition. What we keep are the two properties needed here:
	every component has at most one vertex of degree greater than \(d\) in
	\(G\), and \(G'[C]=G[C]\).
	
\begin{algorithm}[H]
    \SetAlgoLined
    \caption{\textbf{Decompose}}
    \label{GP}
    \KwIn{an outerplanar graph $G$, a threshold $d$, a proximity parameter $\varepsilon$, and a fixed seed $\mathbf R$}
    \KwOut{a spanning subgraph $G'_{\mathbf R}\subseteq G$ whose components contain at most one vertex of $V^h:=\{v:\deg_G(v)>d\}$}
    Let $V^h:=\{v\in V(G):\deg_G(v)>d\}$ and $V^\ell:=V(G)\setminus V^h$\;
    Let $G_1$ be the spanning subgraph obtained from $G$ by deleting all edges with both endpoints in $V^h$\;
    Let $\widetilde G:=(V(G),E(G[V^\ell]))$, so each high-degree vertex is isolated. Apply \Cref{partition} to $\widetilde G$ with parameter $\varepsilon/(4d)$ and seed $\mathbf R$. Discard the singleton high-degree parts, call the remaining partition of $V^\ell$ $\mathcal P^\ell_{\mathbf R}$, and set $s:=q_{\mathrm{part}}(\varepsilon/(4d),d)$\;
    Let $G_2$ be obtained from $G_1$ by deleting all edges whose endpoints lie in different parts of $\mathcal P^\ell_{\mathbf R}$\;
    For each $P\in\mathcal P^\ell_{\mathbf R}$, put $T_P:=\Gamma_G(P)\cap V^h$ and choose the fixed minimum $T_P$-multiway cut in $J_P:=(P\cup T_P,E(G[P])\cup E_G(P,T_P))$. Obtain $G_3$ from $G_2$ by deleting all these cuts\;
    For each connected component $C$ of $G_3$, restore every original edge of $G[C]$, yielding the output $G'_{\mathbf R}$\;
\end{algorithm}
	
	\begin{lemma}[cf. Theorem 3.5 in \cite{BKN16}]
		\label{lem:structural}
		Let $\varepsilon\in(0,1)$, let $d\ge2$, and let $G$ be an outerplanar graph.  Let $s:=q_{\mathrm{part}}(\varepsilon/(4d),d)$.
		
		Assume that $d\ge \frac{240}{\varepsilon}\log(2s+1)$.  For every fixed seed $\mathbf R$, \Cref{GP} deterministically produces a spanning subgraph $G'_{\mathbf R}\subseteq G$ by restoring the original edges internal to the components of $G_3$, such that
		\begin{enumerate}
			\item $G'_{\mathbf R}[C]=G[C]$ for every component $C$ of $G'_{\mathbf R}$; and
			\item every component $C$ of $G'_{\mathbf R}$ satisfies $|C\cap V^h|\le1$, where $V^h=\{v:\deg_G(v)>d\}$.
		\end{enumerate}
		Also, with probability at least $2/3$ over $\mathbf R$, $|E(G)\setminus E(G'_{\mathbf R})|\le \varepsilon n$.
		
		By \Cref{partition}, one may take
		\(s=\operatorname{poly}(d^2/\varepsilon)\), because the oracle is
		invoked with cut parameter \(\varepsilon/(4d)\). Therefore, for a sufficiently large constant \(N>0\), independent of
		$G,H,\varepsilon$ and depending only on the fixed outerplanar partition oracle, the choice $d\ge \frac{N}{\varepsilon^2}$ satisfies the above condition.
	\end{lemma}

	\begin{remark}
		The graphs $G_3$ and $G'_{\mathbf R}$ have the same connected components, and the first two conclusions of \Cref{lem:structural} hold for every seed, not merely on the cut-bound event. In particular, every connected induced copy found in $G'_{\mathbf R}$ is an induced copy in $G$.
		
		A connected candidate in $G'_{\mathbf R}$ contains at most one vertex whose degree in the original graph $G$ exceeds $d$, and hence all adjacencies inside the candidate can be checked by scanning a low-degree endpoint.  For disconnected $H$, deleting intercomponent edges may create new induced copies, which motivates the structural analysis below.
	\end{remark}
	
	\subsection{Local access to the decomposition}
	\label{sec:local-simulation}
	
	The preceding subsection defines the decomposed graph globally. This
	subsection shows how the tester can query that graph without constructing it.
	Its main result is \Cref{simulate}: all calls made with one fixed seed refer to
	the same graph \(G'_{\mathbf R}\), and each call uses
	\(O(dN_{\mathrm{part}})\) threshold-list queries.
	
	The idea is to compute a component label from the local partition part and
	the fixed minimum multiway cut. Two adjacent vertices belong to the same component
	of \(G'_{\mathbf R}\) exactly when their labels agree. This gives the
	procedure \Cref{Local}. After proving its guarantee, the connected-case
	algorithm in the next section will use it whenever it explores the implicit
	graph \(G'_{\mathbf R}\).
	
	Fix the seed $\mathbf R$ sampled for this execution. Consistency of the partition oracle holds for every seed, and hence all local calls refer to one partition $\mathcal P_{\mathbf R}$ and one graph $G'_{\mathbf R}$.  The probabilistic cut-bound event is irrelevant to this consistency statement.
	
	The partition oracle is run on the bounded-degree spanning graph $\widetilde G=(V(G),E(G[V^\ell]))$, $
	V^\ell=\{v:\deg_G(v)\le d\}$. Thus, a high-degree vertex is an isolated singleton of $\widetilde G$, and a
	uniformly random vertex is obtained by choosing a uniform identifier in $[n]$.
	A neighbor query to $\widetilde G$ is answered as follows.  If the
	queried vertex is high-degree, return $\perp$.  Otherwise obtain its complete
	list from $\mathsf{TL}_d$, classify its at most $d$ neighbors using threshold-list queries, retain the
	low-degree ones, and sort them by identifier. This uses \(O(d)\)
	threshold-list queries and gives one exact, fixed adjacency-list
	representation of \(\widetilde G\).
	
	For a low-degree vertex $v$, let $P:=P(v)$ be the part returned by
	$\mathbf{findPartition}(v,\mathbf R)$. Using the exact complete lists of
	the vertices in $P$, we reconstruct the bounded graph
	$J_P=(P\cup T_P,E(G[P])\cup E_G(P,T_P))$, where
	$T_P:=\Gamma_G(P)\cap V^h$. This is the same $J_P$ as in \Cref{GP}; it is the subgraph of $G_2$ induced by $P\cup T_P$. In particular, it contains no high--high edges. Delete from $J_P$ the same fixed minimum $T_P$-multiway cut used to obtain $G_3$.  The component of $v$ in the resulting graph contains at most one high-degree root.  Define its \emph{component label} to be that root if one is present, and otherwise the ordered pair consisting of $P(v)$ and the vertex set of the component of \(v\) after the cut.  A high-degree vertex $a$ has label $a$.
	
	Two vertices have the same label exactly when they belong to the same connected component of $G_3$.  Indeed, residual pieces from different low-degree parts can meet only at a common high-degree root.  The restoration step adds only edges internal to an existing component, and hence it does not change the component partition.  Finally, because $G'_{\mathbf R}[C]=G[C]$ for every component $C$, an edge $uv\in E(G)$ is retained in $G'_{\mathbf R}$ if and only if its endpoints have the same label.
	
	\begin{algorithm}[H]
		\SetAlgoLined
		\caption{\textbf{IsRetained}}
		\label{Local}
		\KwIn{an edge $v_1v_2\in E(G)$, threshold $d$, proximity parameter $\zeta$, and fixed seed $\mathbf R$}
		\KwOut{\textnormal{\textsc{Yes}} if $v_1v_2\in E(G'_{\mathbf R})$, and \textnormal{\textsc{No}} otherwise}
		Call $\mathsf{TL}_d(v_1)$ and $\mathsf{TL}_d(v_2)$\;
		\If{both answers are $\mathsf{High}$}{\Return{\textnormal{\textsc{No}}}}
		Compute the component label of each low-degree endpoint as above, using the partition-oracle parameter $\zeta/(4d)$; use a high-degree endpoint itself as its label\;
		\Return{\textnormal{\textsc{Yes}} exactly when the two labels agree.}
	\end{algorithm}
	
	\begin{lemma}[Exact local simulation]\label{simulate}
		For every fixed seed \(\mathbf R\), every answer of \Cref{Local} is correct
		with respect to the same graph $G'_{\mathbf R}
		=
		\mathbf{Decompose}_{\mathbf R}(G,d,\zeta)$.
		Let $N_{\mathrm{part}}
		:=
		q_{\mathrm{part}}\!\left(\frac{\zeta}{4d},d\right)$
		where \(q_{\mathrm{part}}\) is the partition-oracle bound from
		\Cref{partition}. Then one invocation of \Cref{Local} uses $N_{\mathrm{sim}}
		:=
		O\!\left(dN_{\mathrm{part}}\right)$
		threshold-list queries.
	\end{lemma}

	\begin{proof}
		Correctness follows from the component-label characterization above and is
		simultaneous for all calls using $\mathbf R$.  A part has size at most
		$N_{\mathrm{part}}$.  Simulating all bounded-degree-oracle calls made by
		$\mathbf{findPartition}$ and reconstructing the returned parts uses
		$O(dN_{\mathrm{part}})$ threshold-list queries. The chosen cut and
		labels are then computed without additional graph access.
	\end{proof}

	\section{The connected case: algorithm and analysis}
	\label{sec:induced-finder}
	This section gives the algorithm and analysis for a connected forbidden graph.
	Its main result is \Cref{thm:test-connected}.

	\begin{theorem}[Connected tester with threshold-list access]
		\label{thm:test-connected}
		Let \(H\) be a fixed connected graph. Then
		there is a threshold \(d=(1/\varepsilon)^{O(1)}\) for which
		induced-\(H\)-freeness in outerplanar graphs is testable with one-sided
		error with threshold-list access. The tester uses exact
		\(\mathsf{TL}_d\) queries and direct random-neighbor samples at
		high-degree roots, and its total query complexity is
		\((1/\varepsilon)^{O_H(1)}\).
	\end{theorem}
	
	The proof has three parts. First, we show that a graph far from
	induced-\(H\)-free contains many induced copies that local sampling can
	find with a probability bounded below independently of \(n\). Second, we turn these facts into the connected-copy finder
	in \Cref{thm:find-one-induced}. Third, we apply the finder to the decomposed
	graph and prove \Cref{thm:test-connected}. After that proof, we give a
	sequential form of the finder that will be used only in the disconnected
	case.
	
	\subsection{Parameter setting and the threshold-list query}
	
	\paragraph*{Parameter setting.}
	Let \(N\) be the absolute constant from \Cref{lem:structural}, and set $d:=\left\lceil\frac{100N}{\varepsilon^2}\right\rceil$. 
	This choice lets us apply \Cref{lem:structural} with proximity parameter
	\(\varepsilon/10\).
	Throughout this section, one execution of the tester uses one fixed
	random seed \(\mathbf R\) for the decomposition, and we write $G'
	:=
	G'_{\mathbf R}
	=
	\mathbf{Decompose}_{\mathbf R}
	\left(
	G,d,\frac{\varepsilon}{10}
	\right)$. All invocations of \Cref{Local} during this execution use
	the same seed \(\mathbf R\) and the decomposition parameter
	$\zeta=\varepsilon/10$; this suppressed parameter is also used by every
	call to $\mathbf{ExploreLow}$.
	
	Recall that, when
	\(\mathbf{Decompose}(G,d,\varepsilon/10)\) is executed, the partition
	oracle on the spanning graph $\widetilde G$ is invoked with parameter
	$\frac{\varepsilon}{40d}.$
	Let $N_{\mathrm{part}}
	:=q_{\mathrm{part}}
	\left(\frac{\varepsilon}{40d},d\right)$. By \Cref{partition}, $N_{\mathrm{part}}
	=
	\operatorname{poly}
	\left(
	\frac{d}{\varepsilon/(40d)}
	\right)
	=
	\operatorname{poly}
	\left(
	\frac{d^2}{\varepsilon}
	\right)$. By \Cref{simulate}, one membership query for an edge of the implicit
	graph \(G'\) can be simulated using
	$N_{\mathrm{sim}} = O(dN_{\mathrm{part}})$
	threshold-list queries.

	Throughout the finder analysis, \(F\) denotes a fixed connected
	nontrivial graph. Let $r:=|V(F)|,\,m:=|E(F)|$ and $\Delta:=\Delta(F)$. 
	Recall that
	\(V^h:=\{v\in V(G):\deg_G(v)>d\}\) and
	\(V^\ell:=V(G)\setminus V^h\).
	
	The finder analysis works for any fixed seed, whether or not the
	decomposition deletes few edges. By \Cref{lem:structural}, every
	connected component of \(G'\) contains at most one vertex of
	\(V^h\), and for every connected component \(C\) of \(G'\),
	$G'[C]=G[C].$
	Also, by \Cref{simulate}, every invocation of
	\Cref{Local} using the same seed \(\mathbf R\) correctly
	decides whether a queried edge of \(G\) belongs to this same graph
	\(G'\).
	
	\paragraph*{Exact graph queries.}
	We first analyze the tester with threshold-list access as defined in \Cref{pre}.
	Thus, every query \(\mathsf{TL}_d(v)\) is answered exactly. The tester may
	also take uniform vertex samples and uniformly
	sample a neighbor of a high-degree vertex. Thus, phrases such as
	``determine whether \(v\in V^\ell\)'' and ``inspect the neighborhood of a
	low-degree vertex'' mean one threshold-list query. These operations are only
	shorthand for adjacency-list access. In the proof of
	\Cref{thm:sparse-main}, every threshold-list query and every high-degree
	neighbor sample is implemented by exact degree and indexed-neighbor
	queries. After that proof, \Cref{sec:rn-recovery} shows how to recover the
	same bounded-degree lists with high probability using only random-neighbor
	queries.
	
	The next two subsections prove the finder theorem. We first find many edge-disjoint copies.
	We then discard some copies so that a random-neighbor query at any vertex
	used by the remaining family has a positive probability of finding one
	of its edges. Bounded
	exploration then recovers low-degree pieces, and a compatibility argument
	joins branches around a high-degree root. We state the finder after these
	ingredients and then analyze its two search branches.
	
	\subsection{Structural facts for graphs far from induced-freeness}
	
	\paragraph*{Many edge-disjoint induced copies.}
	An \emph{edge-disjoint family}, also called an edge packing, may share
	vertices but never edges. To obtain such a family, take a maximal one.
	If it were small, isolating all of its low-degree vertices would destroy
	every induced copy at low cost. Isolation is important here: deleting
	only the edges of the selected copies could create new induced copies.

	\begin{lemma}[Many edge-disjoint induced copies]
		\label{lem:induced-edge-packing-new}
		If \(G'\) is \(\alpha\)-far from being induced-\(F\)-free, then
		\(G'\) contains at least
		$\frac{\alpha n}{dr}$
		pairwise edge-disjoint induced copies of \(F\).
	\end{lemma}
	
	\begin{proof}
		Let \(\mathcal Q\) be a maximal family of pairwise edge-disjoint
		induced copies of \(F\) in \(G'\), and define
		$q:=|\mathcal Q|$.
		For every \(Q\in\mathcal Q\), define
		$L(Q):=V(Q)\cap V^\ell,$
		and let
		$W:=\bigcup_{Q\in\mathcal Q}L(Q).$
		Since every copy has \(r\) vertices,
		$|W|\le rq.$
		
		Delete from \(G'\) every edge incident with a vertex of \(W\).
		Every vertex of \(W\) has degree at most \(d\) in \(G\), and hence it has degree at most \(d\)
		also in \(G'\). Therefore, at most
		$d|W|\le drq$
		edges are deleted.
		
		We claim that the resulting graph is induced-\(F\)-free.
		Suppose otherwise, and let \(R\) be an induced copy of \(F\) in
		the resulting graph. Every vertex of \(W\) is isolated after the
		deletions. Since \(F\) is connected and nontrivial,
		$V(R)\cap W=\emptyset.$
		No adjacency whose two endpoints lie outside \(W\) has been
		changed. Thus, \(R\) was already an induced copy of \(F\) in \(G'\).
		
		By maximality of \(\mathcal Q\), the copy \(R\) must share an edge
		\(uv\) with some \(Q\in\mathcal Q\). Since \(Q\) is connected, it
		lies in a single connected component of \(G'\). By
		\Cref{lem:structural}, this component contains at most one vertex of
		\(V^h\). Therefore, at least one endpoint of \(uv\), say \(u\), belongs
		to \(V^\ell\). Hence,
		$u\in L(Q)\subseteq W,$
		contradicting \(V(R)\cap W=\emptyset\).
		
		Consequently, deleting at most \(drq\) edges makes \(G'\)
		induced-\(F\)-free. Since \(G'\) is \(\alpha\)-far from being
		induced-\(F\)-free,
		$drq\ge \alpha n,$
		and hence
		$|\mathcal Q| \ge \frac{\alpha n}{dr}.$
	\end{proof}

	\paragraph*{Keeping enough edges for random sampling.}
	
	By \Cref{lem:induced-edge-packing-new}, there exists a pairwise
	edge-disjoint family \(\mathcal Q\) of induced copies of \(F\) in \(G'\)
	with $|\mathcal Q|
	\ge
	\frac{\alpha n}{dr}$. But the existence of many such copies alone does not guarantee that
	a random-neighbor query at a vertex participating in these copies is likely
	to sample an edge belonging to one of them. We keep a large subfamily \(\mathcal Q^\star\) in which, at every
	vertex used by a copy, a positive fraction of the incident edges belongs
	to the family. We do this by repeatedly removing copies through any
	vertex where this fraction is too small. The sum of all degrees in an
	outerplanar graph is less than \(4n\), so these removals discard only a
	small fraction of the copies.

	For a subfamily \(\mathcal A\subseteq\mathcal Q\), let $U(\mathcal A):=\left(V(G),\bigcup_{Q\in\mathcal A}E(Q)\right)$ 
	be the graph formed by the edges of its copies. Set
	\(\beta:=\alpha/(dr)\) and \(\rho:=\beta/12\).
	Thus, for a nonisolated vertex \(v\),
	\(\deg_{U(\mathcal A)}(v)/\deg_G(v)\) is the probability that a
	random-neighbor query at \(v\) returns an edge of this family.

	\begin{lemma}[A subfamily with enough edges at every used vertex]
		\label{lem:induced-pruning}
		Let \(\mathcal Q\) be a family of at least \(\beta n\) pairwise
		edge-disjoint induced copies of \(F\) in \(G'\).
		There exists a subfamily
		\(\mathcal Q^\star\subseteq\mathcal Q\) with $|\mathcal Q^\star|
		\ge
		\frac{\beta n}{2}$
		such that every nonisolated vertex \(v\) of
		\(U(\mathcal Q^\star)\) satisfies $\deg_{U(\mathcal Q^\star)}(v)
		>
		\rho\deg_G(v)$.
	\end{lemma}

	\begin{proof}
        If $n=0$, take $\mathcal Q^\star=\emptyset$; the conclusions are immediate. Assume $n\ge1$.
		Initially set
		$\mathcal Q_0:=\mathcal Q.$
		Given the current family \(\mathcal Q_t\), write
		$U_t:=U(\mathcal Q_t).$
		If there exists a nonisolated vertex \(v\) of \(U_t\) satisfying $\deg_{U_t}(v)
		\le
		\frac{\beta}{12}\deg_G(v)$,
		delete from \(\mathcal Q_t\) every copy containing \(v\), and denote
		the resulting family by \(\mathcal Q_{t+1}\).
		
		Since the members of \(\mathcal Q_t\) are pairwise edge-disjoint and
		\(F\) is connected and nontrivial, every copy of \(\mathcal Q_t\)
		containing \(v\) contains at least one edge incident with \(v\).
		Also, distinct copies use distinct edges incident with \(v\).
		Therefore, the number of copies deleted when processing \(v\) is at most $\deg_{U_t}(v)
		\le
		\frac{\beta}{12}\deg_G(v)$.
		
		After all copies containing \(v\) are deleted, the vertex \(v\) is
		isolated in \(U_{t+1}\). Since later steps only delete more
		copies, \(v\) remains isolated and is never processed again.
		
		Consequently, each vertex of \(G\) is processed at most once, and the
		total number of deleted copies is at most $\frac{\beta}{12}
		\sum_{v\in V(G)}\deg_G(v)$. Since \(G\) is outerplanar,
		$\sum_{v\in V(G)}\deg_G(v)<4n.$
		Therefore, fewer than
		$\frac{\beta n}{3}$
		copies are deleted.
		
		Let \(\mathcal Q^\star\) be the family when the procedure terminates.
		By \Cref{lem:induced-edge-packing-new},
		$|\mathcal Q|\ge\beta n,$
		and hence $|\mathcal Q^\star|
		\ge
		\beta n-\frac{\beta n}{3}
		=
		\frac{2\beta n}{3}
		\ge
		\frac{\beta n}{2}$.
		
		Finally, at termination there is no nonisolated vertex \(v\) of
		\(U(\mathcal Q^\star)\) satisfying $\deg_{U(\mathcal Q^\star)}(v)
		\le
		\frac{\beta}{12}\deg_G(v)$. Thus, every nonisolated vertex \(v\) satisfies
		$\deg_{U(\mathcal Q^\star)}(v)
		>
		\frac{\beta}{12}\deg_G(v)$, as required.
	\end{proof}

	We shall repeatedly use the following simple observation.
	
	\begin{lemma}[Few edge-disjoint copies through a low-degree vertex]
		\label{lem:low-multiplicity}
		Let \(\mathcal A\) be a pairwise edge-disjoint family of connected
		nontrivial subgraphs of \(G'\). Every vertex \(v\in V^\ell\) belongs to at
		most \(d\) members of \(\mathcal A\).
	\end{lemma}
	
	\begin{proof}
		Every member of \(\mathcal A\) containing \(v\) uses at least one
		edge incident with \(v\). Since the members of \(\mathcal A\) are
		edge-disjoint, these incident edges are distinct. Since
		$\deg_{G'}(v)\le\deg_G(v)\le d,$
		at most \(d\) members of \(\mathcal A\) contain \(v\).
	\end{proof}

	\paragraph*{Compatibility of branches at a high-degree root.}
	The high-degree case requires one additional ingredient. Suppose that many
	edge-disjoint induced copies of \(F\) share the same high-degree vertex \(a\),
	with the same vertex \(z\in V(F)\) mapped to \(a\). A random-neighbor
	exploration at \(a\) need not recover all branches from the same copy:
	different queries may lead to branches belonging to different copies.
	
	Write \(B_1,\ldots,B_c\) for the connected components of \(F-z\). We want to choose, independently for each \(B_j\), its image from one
	of the copies sharing \(a\), and still obtain an induced copy of \(F\).
	The only possible obstruction is that two selected branch images intersect
	or have an edge between them. Since all branch vertices are low-degree and
	the copies are pairwise edge-disjoint, \Cref{lem:low-multiplicity} implies
	that each fixed branch is incompatible with only \(O_{F,d}(1)\) choices of
	any other branch type. The following lemma shows that independently selected
	branches are compatible with constant probability.
	
	For the lemma, fix \(a\in V^h\) and a family
	\(\mathcal R=\{Q_1,\ldots,Q_t\}\) of pairwise edge-disjoint induced
	\(F\)-copies in \(G'\), all containing \(a\). Fix isomorphisms
	\(\phi_s:F\to Q_s\) that map the same vertex \(z\in V(F)\) to \(a\).
	As above, \(B_1,\ldots,B_c\) are the components of \(F-z\).

	\begin{lemma}[Compatibility of independently sampled branches]
		\label{lem:compatible-branches}
		Independently choose \(I_1,\ldots,I_c\) uniformly from \([t]\), and set $A_j:=\phi_{I_j}(V(B_j))$ where $j\in[c]$.
		There exists a constant \(\kappa=\kappa(F,d)>0\) such that, with
		probability at least \(\kappa\), the sets \(A_1,\ldots,A_c\) are
		pairwise disjoint and no edge of \(G'\) joins \(A_j\) and \(A_\ell\)
		for distinct \(j,\ell\). Therefore, on this event, $G'\left[\{a\}\cup\bigcup_{j=1}^c A_j\right]\cong F$.
	\end{lemma}
	
	\begin{proof}
		For $A=\phi_s(V(B_j))$ and $
		A'=\phi_{s'}(V(B_\ell))$ where $j\neq\ell$,
		call \(A\) and \(A'\) \emph{incompatible} if either
		\(A\cap A'\neq\emptyset\) or there is an edge of \(G'\) joining a
		vertex of \(A\) to a vertex of \(A'\).
		
		Since \(F\) is connected, every \(Q_s\) lies in the connected component
		of \(G'\) containing \(a\). This component contains at most one
		high-degree vertex, which is \(a\). Thus, every vertex of every branch
		belongs to \(V^\ell\).
		
		Fix a branch \(A=\phi_s(V(B_j))\), and define $Z_A:=
		\bigl(A\cup\Gamma_{G'}(A)\bigr)\cap V^\ell$.
		Since \(|A|\le r-1\) and every vertex of \(A\) has degree at most \(d\), we have
		$|Z_A|\le r(d+1)$.
		
		Fix another branch type \(B_\ell\). If
		\(\phi_{s'}(V(B_\ell))\) is incompatible with \(A\), then it contains
		at least one vertex of \(Z_A\). By \Cref{lem:low-multiplicity}, every
		vertex of \(Z_A\) belongs to at most \(d\) members of \(\mathcal R\).
		Therefore, the number of indices \(s'\in[t]\) for which
		\(\phi_{s'}(V(B_\ell))\) is incompatible with \(A\) is at most $C_0:=rd(d+1)$.
		
		Define $T_0:=
		\max\left\{
		2,\,
		2\binom{r-1}{2}C_0
		\right\}$.
		Suppose first that \(t\ge T_0\). For every two distinct branch types
		\(B_j,B_\ell\), conditioned on the first selected branch, the probability
		that the second selected branch is incompatible with it is at most
		\(C_0/t\). Thus,
		\[
		\Pr[
		\text{some two selected branches are incompatible}
		]
		\le
		\binom{c}{2}\frac{C_0}{t}
		\le
		\frac12.
		\]
		Therefore, the selected branches are pairwise compatible with probability at
		least \(1/2\).
		
		Now suppose that \(t<T_0\). The event $I_1=\cdots=I_c$
		guarantees that all selected branches come from the same induced copy
		\(Q_s\), and hence they are pairwise compatible. Its probability is $t^{1-c}\ge T_0^{2-r}$
		since \(c\le r-1\).
		
		Thus, the selected branches are pairwise compatible with probability at
		least $$\kappa:=\min\left\{\frac12,\,T_0^{2-r}\right\}>0.$$  
		On this event, each selected branch has the correct internal structure
		and the correct adjacency to \(a\), while compatibility excludes
		intersections and edges between distinct branches. Thus, $G'\left[
		\{a\}\cup\bigcup_{j=1}^c A_j
		\right]
		\cong F$.
	\end{proof}

	\paragraph*{Exploration from low-degree vertices.}
	The next subroutine explores only a bounded number of vertices, even
	when it encounters a vertex of very large degree.
	For \(x\in V^\ell\), the procedure
	\(\mathbf{ExploreLow}(x,r,\mathbf R,S)\) performs breadth-first search
	from \(x\), following edges of \(G'-S\) for at most \(r\) steps.
	Here \(S\) is an explicitly known set of excluded vertices, empty unless
	specified otherwise, and \(x\notin S\). No vertex in \(S\) is recorded
	or expanded, and every edge incident with \(S\) is ignored.
	It expands low-degree vertices at search levels \(0,\ldots,r-1\).
	Vertices at level \(r\), and high-degree vertices at any level, are
	recorded but not expanded. Thus, a search path never has a high-degree
	internal vertex. Whenever a low-degree vertex \(v\) is expanded, use the
	complete list returned by $\mathsf{TL}_d(v)$.
	For every discovered edge \(vw\in E(G)\), invoke
	\Cref{Local}, using the fixed seed \(\mathbf R\), and retain
	\(vw\) exactly when \Cref{Local} declares that
	$vw\in E(G').$
	If \(w\in V^h\), record \(w\), but do not expand it.
	
	Let
	$S_r(d):=\sum_{j=0}^{r-1}d^j.$
	At most \(S_r(d)\) low-degree vertices are expanded and at most
	\(dS_r(d)\) edges of \(G\) are inspected. Define $N_{\mathrm{ExploreLow}}(F,d,\varepsilon)
	:=
	S_r(d)+dS_r(d)N_{\mathrm{sim}}$. 
	This bounds the number of threshold-list queries used by one exploration. In
	particular, $N_{\mathrm{ExploreLow}}(F,d,\varepsilon)
	=
	O_F(d^rN_{\mathrm{sim}})$.

	\subsection{The connected induced-copy finder}

	We now combine the preceding ingredients into a local procedure for finding an
	induced copy of \(F\). By
	\Cref{lem:induced-edge-packing-new,lem:induced-pruning}, whenever \(G'\) is
	far from being induced-\(F\)-free, it contains a linear-size family
	\(\mathcal Q^\star\) of edge-disjoint induced \(F\)-copies such that every
	vertex used by this family has at least a \(\rho\)-fraction of its
	incident edges in the family.
	
	We distinguish two possibilities. If many copies in \(\mathcal Q^\star\)
	contain only low-degree vertices, then a uniformly sampled vertex lands in
	one of them with constant probability, and \(\mathbf{ExploreLow}\) exposes
	the entire copy. Otherwise, many copies contain a unique high-degree vertex.
	Starting from a sampled low-degree vertex, we first reach this high-degree
	root and then use random-neighbor queries at the root to sample candidate
	branches. The degree bound for this subfamily guarantees that such a query hits a
	useful branch with probability bounded below in terms of \(F,\alpha,d\), while
	\Cref{lem:compatible-branches} shows that independently sampled branches can
	be combined into an induced copy of \(F\) with positive probability.
	
	The algorithm below carries out these two searches simultaneously. Every
	candidate it produces is finally verified in the original graph \(G\), and therefore
	only genuine induced copies are returned.

	\begin{algorithm}[H]
		\SetAlgoLined
		\caption{\textbf{FindInduced}}
		\label{alg:find-induced-one}
		\KwIn{$G$, a fixed connected nontrivial graph $F$, threshold $d$, seed $\mathbf R$, and a known excluded set $S$ (default $\emptyset$)}
		\KwOut{an induced copy of $F$ in $G'-S$, or \textnormal{\textsc{Fail}}}
		Set \(r:=|V(F)|\)\;
		\If{$n<r$}{\Return{\textnormal{\textsc{Fail}}}}
		Choose $x\in V(G)$ uniformly at random\;
		\If{$x\in S$}{\Return{\textnormal{\textsc{Fail}}}}
		Call $\mathsf{TL}_d(x)$\;
		\If{the answer is $\mathsf{High}$}{\Return{\textnormal{\textsc{Fail}}}\;}
		Run $\mathbf{ExploreLow}(x,r,\mathbf R,S)$\;
		Let $\mathcal E$ be the explored subgraph of $G'-S$\;
		\If{a vertex $a\in V^h$ is encountered}{
			Make $r-1$ fresh calls to $\mathsf{RN}(a)$\;
			\ForEach{sampled neighbor $y\notin S$}{
				Call $\mathsf{TL}_d(y)$\;
				Use \Cref{Local} to test whether $ay\in E(G')$\;
				\If{$ay\in E(G')$ and $y$ is $\mathsf{Low}$}{
					Add the edge $ay$ to $\mathcal E$\;
					Run $\mathbf{ExploreLow}(y,r,\mathbf R,S)$\;
					Add its explored edges to $\mathcal E$\;
				}
			}
		}
		Enumerate injective maps \(\phi:V(F)\to V(\mathcal E)\) for which
		every edge of \(F\) maps to an edge of \(\mathcal E\)\;
		
		\ForEach{candidate $C$}{
			Call $\mathsf{TL}_d$ on all vertices of $C$\;
			Discard $C$ if two answers are $\mathsf{High}$\;
			Otherwise keep $C$ exactly when $G[V(C)]\cong F$\;
		}
		\If{some candidate passes verification}{\Return{$C$}}
		\Return{\textnormal{\textsc{Fail}}}
	\end{algorithm}
	
	The final verification is exact. Indeed, the algorithm discards every
	candidate containing two high-degree vertices. Thus, every unordered pair
	of vertices in a surviving candidate has at least one low-degree endpoint,
	whose complete neighborhood is known from \(\mathsf{TL}_d\). We can
	determine every adjacency in \(G[V(C)]\) exactly using \(O_F(1)\)
	threshold-list queries. Consequently, with exact threshold-list answers,
	\Cref{alg:find-induced-one} never returns a false induced copy. The random-neighbor implementation will have the same guarantee
	whenever all its recovered lists are correct; see \Cref{sec:rn-simulation}.
	Also, every candidate is connected in \(G'-S\). Since each component
	of \(G'\) induces the same graph in \(G'\) and \(G\), every returned
	copy is induced in \(G'-S\) as well as in \(G\).

	\begin{samepage}
	\begin{theorem}[Finding one connected induced copy]
		\label{thm:find-one-induced}
		Let \(F\) be a fixed connected nontrivial graph and let
		$\alpha\in(0,1)$. Fix an integer $d\ge100N/\varepsilon^2$,
		an arbitrary seed $\mathbf R$, and a known set $S\subseteq V(G)$.
		Let $G':=\mathbf{Decompose}_{\mathbf R}(G,d,\varepsilon/10)$.
		Suppose that $G'-S$
		requires more than $\alpha n$ edge edits to become induced-\(F\)-free,
		where $n=|V(G)|$. Degrees and high/low classifications always refer
		to the original graph $G$.
		
		For every \(\delta\in(0,1)\), there is a randomized procedure
		which outputs an induced copy of \(F\) in $G'-S$ with probability at least
		\(1-\delta\), using $N_{\mathrm{find}}(F,\alpha,\delta)
		=
		O\left(
		\frac{\log(2/\delta)}{p_F(\alpha,d)}
		\left(
		d^rN_{\mathrm{sim}}
		+
		d^{r^2}
		+
		r
		\right)
		\right)$
		threshold-list and direct random-neighbor queries, where
		\(p_F(\alpha,d)>0\) depends only on
		\(F,\alpha,d\). 
        
        More precisely, $p_F(\alpha,d)
		\ge
		\alpha^{O_F(1)}d^{-O_F(1)}$. 
		Thus, if \(d=(1/\varepsilon)^{O(1)}\),
		\(\alpha\ge\varepsilon^{O_F(1)}\), and
		\(\delta\ge\varepsilon^{O_F(1)}\), then $N_{\mathrm{find}}(F,\alpha,\delta)
		=
		\left(\frac{1}{\varepsilon}\right)^{O_F(1)}$. 
		Here \(N_{\mathrm{sim}}\) denotes the query bound from
		\Cref{simulate}. Thus, \(N_{\mathrm{find}}\) suppresses the
		decomposition parameters \(d\) and \(\varepsilon\), which determine
		\(N_{\mathrm{sim}}\).
		The query bound holds on every input, and any returned copy is genuine
		even without the farness assumption. The success guarantee also holds
		conditional on any previous history that fixes $S$ and $\mathbf R$,
		provided the finder uses fresh randomness.
	\end{theorem}
	\end{samepage}
	
	\begin{proof}
		We first prove the statement for $S=\emptyset$, then explain why
		the same argument applies to every residual graph $G'-S$. By
		\Cref{lem:induced-edge-packing-new,lem:induced-pruning}, there is a
		family $\mathcal Q^\star$ of at least $\beta n/2$ pairwise
		edge-disjoint induced copies of \(F\), where $\beta=\alpha/(dr)$. Let
		$U:=U(\mathcal Q^\star)$ be their edge-union graph.  By
		\Cref{lem:induced-pruning}, $\deg_U(v)
		>
		\rho\deg_G(v)$ where $
		\rho:=\frac{\beta}{12}$,
		at every nonisolated vertex \(v\).
		
		Partition $\mathcal Q^\star
		=
		\mathcal Q^\ell
		\cup
		\mathcal Q^h$, where \(\mathcal Q^\ell\) consists of copies containing no
		high-degree vertex and \(\mathcal Q^h\) consists of the remaining
		copies. At least one of these two families has size at least
		$\frac{\beta n}{4}$.
		
		\medskip
		\noindent
		\textbf{Case 1:}
		\(\boldsymbol{|\mathcal Q^\ell|\ge\beta n/4}\).
		
		Let
		$L := \bigcup_{Q\in\mathcal Q^\ell}V(Q).$
		By \Cref{lem:low-multiplicity},
		$r|\mathcal Q^\ell| \le d|L|.$
		Thus,
		$|L| \ge \frac{r\beta}{4d}n.$
		A uniformly sampled vertex \(x\) belongs to \(L\) with probability
		at least
		$p_\ell := \frac{r\beta}{4d}.$
		
		Condition on \(x\in V(Q)\) for some
		\(Q\in\mathcal Q^\ell\). Every vertex of \(Q\) is low-degree.
		Since \(F\) is connected on \(r\) vertices, every vertex of \(Q\)
		is at distance at most \(r-1\) from \(x\) inside \(Q\). Therefore,
		\(\mathbf{ExploreLow}(x,r,\mathbf R)\) exposes every edge of \(Q\),
		and the exact verification step accepts it.
		
		Consequently, one execution succeeds in this case with probability at least
		\(p_\ell\).
		
		\medskip
		\noindent
		\textbf{Case 2:}
		\(\boldsymbol{|\mathcal Q^h|\ge\beta n/4}\).
		
		Every \(Q\in\mathcal Q^h\) contains exactly one high-degree
		vertex. For every \(Q\in\mathcal Q^h\), fix an isomorphism
		$\phi_Q:F\longrightarrow Q.$
		For \(a\in V^h\), let
		$\mathcal Q^h(a)$
		be the copies whose unique high-degree vertex is \(a\).
		
		Partition \(\mathcal Q^h(a)\) according to the vertex of \(F\)
		mapped to \(a\). Let \(\mathcal R_a\) be a largest class, and let
		\(z_a\in V(F)\) be the common vertex mapped to \(a\) by the copies
		in \(\mathcal R_a\). Then $|\mathcal R_a|
		\ge
		\frac{|\mathcal Q^h(a)|}{r}$.
		Put $\mathcal R
		:=
		\bigcup_{a\in V^h}\mathcal R_a$.
		Then
		$|\mathcal R| \ge \frac{\beta n}{4r}.$
		
		Every member of \(\mathcal R\) has \(r-1\) low-degree vertices.
		If $L_h
		:=
		\bigcup_{Q\in\mathcal R}
		\bigl(V(Q)\cap V^\ell\bigr)$,
		then \Cref{lem:low-multiplicity} gives
		$|L_h| \ge \frac{(r-1)\beta}{4rd}n.$
		Therefore, a uniformly sampled vertex belongs to \(L_h\) with
		probability at least
		$p_h^{(0)} := \frac{(r-1)\beta}{4rd}.$
		
		Condition on this event, and let \(a\) be the high-degree root
		reached from the sampled low-degree vertex. Put
		$t_a:=|\mathcal R_a|.$
		Every copy in \(\mathcal Q^h(a)\) contributes at most \(\Delta\)
		edges incident with \(a\) to \(U\), and hence
		$\deg_U(a) \le \Delta|\mathcal Q^h(a)|.$
		Thus,
		$t_a \ge \frac{\deg_U(a)}{r\Delta}.$
		By the degree-retention property, $\frac{t_a}{\deg_G(a)}
		>
		\frac{\rho}{r\Delta}$.
		Set
		$\sigma:=\frac{\rho}{r\Delta}$.
		
		Let
		$B_1,\ldots,B_c$
		be the connected components of \(F-z_a\).  For each $j\in[c]$,
		fix one vertex $w_j\in\Gamma_F(z_a)\cap V(B_j)$.  Order
		$\mathcal R_a=\{Q_1,\ldots,Q_{t_a}\}$ and write $\phi_s$ for
		the fixed isomorphism to $Q_s$.  The set
		$S_j:=\{\phi_s(w_j):s\in[t_a]\}$
		has size $t_a$: otherwise two copies would use the same edge
		incident with $a$, contradicting edge-disjointness.  Analytically
		designate the first $c$ of the $r-1$ independent samples and require
		that sample $j$ lie in $S_j$.  This event has probability $\left(\frac{t_a}{\deg_G(a)}\right)^c
		\ge \sigma^c\ge\sigma^{r-1}$,
		since $c\le r-1$ and $\sigma\le1$.  Conditioned on this event,
		the indices $I_1,\ldots,I_c$ of the hit copies are independent and
		uniform on $[t_a]$. Each exploration from the sampled neighbor
		reveals the whole corresponding branch and all its edges to $a$.
		\Cref{lem:compatible-branches} implies that the selected branches
		are compatible with probability at least
		$\kappa(F,d).$
		Consequently, one execution succeeds in this case with probability at
		least $p_h
		:=
		\frac{(r-1)\beta}{4rd}
		\left(
		\frac{\rho}{r\Delta}
		\right)^{r-1}
		\kappa(F,d)$.
		
		Define
		$p_F(\alpha,d) := \min\{p_\ell,p_h\}>0.$
		The displayed formulas for \(p_\ell\), \(p_h\), and
		\(\kappa(F,d)\) give
		\(p_F(\alpha,d)\ge
		\alpha^{O_F(1)}d^{-O_F(1)}\).
		Repeating the basic procedure $T
		:=
		\left\lceil
		\frac{\log(2/\delta)}
		{p_F(\alpha,d)}
		\right\rceil$
		times gives failure probability at most \(\delta\).
		
		One initial exploration and at most $r-1$ further explorations expose
		$M=O_F(d^r)$ vertices and edges and use
		$O_F(d^rN_{\mathrm{sim}})$ threshold-list queries.  There are at most
		$M^r=O_F(d^{r^2})$ injective maps of $V(F)$ into the explored
		vertex set, and the exact check of each remaining candidate costs
		$O_F(1)$ threshold-list queries.  At most $r-1$ fresh direct
		random-neighbor calls are made at the high root.  Therefore, one execution uses $O_F\left(
		d^rN_{\mathrm{sim}}
		+
		d^{r^2}
		+
		r
		\right)$
		threshold-list and direct random-neighbor queries. Multiplying by \(T\) proves
		the claimed bound.

		\paragraph*{Excluding a known vertex set.}
		Now fix $S$ and put $X:=G'-S$. Each component of $X$ lies inside
		one component of $G'$. It therefore has at most one vertex of $V^h$
		and induces the same graph in $X$ and in $G$.
		The edge-packing argument applies to $X$: isolating the low-degree
		vertices of a maximal packing costs at most $dr$ per copy, while
		farness is still measured by $\alpha n$. Thus there are at least
		$\beta n$ packed copies in $X$. The pruning argument also applies,
		using the original degrees $\deg_G(v)$ and the bound
		$\sum_{v\in V(G)}\deg_G(v)<4n$. It leaves at least $\beta n/2$
		copies with the same degree-retention guarantee.

		All these copies avoid $S$, and the low-degree multiplicity and
		branch-compatibility bounds are unchanged. Sampling a uniform
		identifier from the original $[n]$ gives the same lower bounds
		$p_\ell$ and $p_h^{(0)}$ on hitting their low-degree vertices.
		Root samples still use $\Gamma_G(a)$, with denominator $\deg_G(a)$;
		samples in $S$ simply do not help the search. On the success events
		analyzed above, exploration in $X$ exposes the same required copies
		or branches. Checking membership in the explicitly stored set $S$
		uses no graph query, so the query bound is unchanged.

		These estimates hold for every fixed $S$ and $\mathbf R$.
		Conditioning on an earlier history only fixes these objects.
		Fresh samples therefore give the same success bound at each later call.
	\end{proof}
	
	\subsection{Proof of the connected tester}
	We now apply the finder to the decomposed graph and prove the main theorem of
	this section.
	
	\begin{proof}[Proof of \Cref{thm:test-connected}]
		If $H=K_1$, reject exactly when $n>0$ (the only induced-$K_1$-free
		graph has no vertices). From now on, assume that \(H\) is nontrivial and
		use the degree threshold \(d\) fixed above.
		One trial samples a fresh seed $\mathbf R$, sets
		$G':=\mathbf{Decompose}_{\mathbf R}(G,d,\varepsilon/10)$, and runs
		the finder from \Cref{thm:find-one-induced}. With probability at least $2/3$, $|E(G)\setminus E(G')|
		\le
		\frac{\varepsilon n}{10}$.
		We claim that if \(G\) is \(\varepsilon\)-far from being
		induced-\(H\)-free, then \(G'\) is
		\(\frac{4\varepsilon}{5}\)-far from being
		induced-\(H\)-free. Otherwise \(G'\) could be made
		induced-\(H\)-free using at most
		$\frac{4\varepsilon n}{5}$
		edge modifications. Together with the decomposition edits, this
		would make \(G\) induced-\(H\)-free using at most $\frac{\varepsilon n}{10}
		+
		\frac{4\varepsilon n}{5}
		=
		\frac{9\varepsilon n}{10}
		<
		\varepsilon n$ modifications, a contradiction.
		
		Apply \Cref{thm:find-one-induced} with $F=H,\,
		\alpha=\frac{4\varepsilon}{5}$,
		and failure probability $1/20$. Reject if an
		induced copy of \(H\) is found, and accept otherwise.
		
		Every returned copy is verified exactly in \(G\) using exact
		low-degree lists. Consequently, an
		induced-\(H\)-free input is always accepted, for every seed. If $G$
		is $\varepsilon$-far, one trial rejects with probability at least
		$(2/3)(19/20)=19/30$.  Run two independent trials with fresh seeds
		and fresh finder randomness and reject if either trial rejects.  The
		resulting rejection probability is at least
		$1-(11/30)^2>2/3$.
		
		Finally, by \Cref{thm:find-one-induced}, each trial uses at most $N_{\mathrm{find}}
		\left(
		H,\frac{4\varepsilon}{5},\frac{1}{20}
		\right)$
		queries. Therefore, the two trials use at most $2N_{\mathrm{find}}
		\left(
		H,\frac{4\varepsilon}{5},\frac{1}{20}
		\right)
		=
		\left(\frac{1}{\varepsilon}\right)^{O_H(1)}$
		queries.
	\end{proof}
	
	\subsection{Finding several vertex-disjoint copies}
	The connected tester needs only one successful call to
	\Cref{thm:find-one-induced}, and therefore the next corollary is not used in the proof
	of \Cref{thm:test-connected}. We record it here because the disconnected
	tester in \Cref{sec:disconnected-tester} must build several vertex-disjoint
	copies of each connected component. The corollary shows that the same finder
	can be applied repeatedly as long as the residual graph stays far from
	induced-freeness. The first alternative in \Cref{far} will give exactly this
	assumption for each required component of a disconnected \(H\).
	
	\begin{corollary}[Repeated search while the remaining graph is far]
		\label{cor:sequential-induced-finder}
		Let \(F\) be a fixed connected nontrivial graph. Fix a seed
		$\mathbf R$ and put $G':=G'_{\mathbf R}$. Suppose that, for every
		collection of at most \(K\) pairwise vertex-disjoint induced copies
		of \(F\) in $G'$, the residual $G'[V(G')\setminus S]$ requires more
		than $\alpha n$ edits to become induced-\(F\)-free, where $S$ is the
		union of the copy vertices and distance is normalized by the original
		$n=|V(G)|$.
		
		Then, using this fixed seed \(\mathbf R\), for every
		\(\delta\in(0,1)\) one can find \(K+1\) pairwise vertex-disjoint
		induced copies of \(F\) with probability at least \(1-\delta\), using
		\[
		N_{\mathrm{seq}}(F,\alpha,K,\delta)
		\le
		(K+1)
		N_{\mathrm{find}}
		\left(
		F,\alpha,\frac{\delta}{K+1}
		\right)
		\]
		threshold-list and direct random-neighbor queries, where
		\(N_{\mathrm{find}}\) is the query bound from
		\Cref{thm:find-one-induced}. As with \(N_{\mathrm{find}}\), the notation
		\(N_{\mathrm{seq}}\) suppresses the decomposition parameters
		\(d\) and \(\varepsilon\).
	\end{corollary}
	
	\begin{proof}
		Suppose that after \(t\le K\) successful stages we have found
		pairwise vertex-disjoint induced copies
		$C_1,\ldots,C_t.$
		Put
		$S_t:=\bigcup_{j=1}^tV(C_j).$
		By assumption,
		$G'[V(G')\setminus S_t]$
		remains \(\alpha\)-far from being induced-\(F\)-free.
		
		Apply \Cref{thm:find-one-induced} with excluded set $S_t$ and
		failure probability $\delta/(K+1)$. Its conditional guarantee
		applies after every history of successful earlier stages. A union
		bound over the $K+1$ stages proves the success probability, and
		summing their query bounds gives the stated complexity.
	\end{proof}

\section{Structural tools for disconnected forbidden graphs}
\label{sec:disconnected-structure}
\suppressfloats[t]

The main result of this section is the exceptional-component dichotomy in
\Cref{far}. On a far input, either every component remains findable after
a bounded family of its copies is excluded, or there is exactly one
exceptional component. In the second case, a witness copy of this component
has a separated region in which all the other components can be found.
This is the structural reason that the connected finder extends to
disconnected forbidden graphs.

For a connected forbidden graph, it is enough to find and check one candidate
copy. As \Cref{fig:disconnected-obstacle} illustrates, a disconnected graph
creates a second task: candidates for its components must also have no edges
between them.  Finding each component separately does not ensure this
condition. There is also a different obstruction: a far input may have
very few vertex-disjoint copies of one component. The example below
explains why this requires the exceptional alternative. We then define
the terms needed to state \Cref{far} precisely and prove it. The next
section turns the two alternatives into a tester.

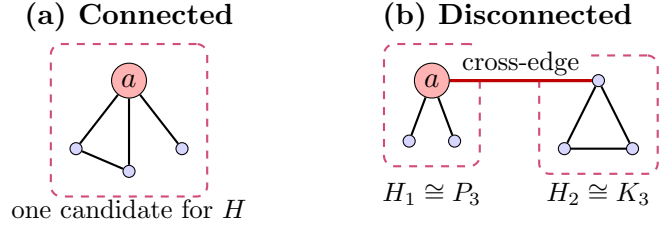
\begin{figure}[t]
	\centering
	\begin{tikzpicture}[
		vertex/.style={circle,draw,fill=blue!16,inner sep=1.6pt},
		high/.style={circle,draw,fill=red!30,inner sep=2pt},
		copy/.style={draw=purple!70,dashed,rounded corners,thick},
		bad/.style={draw=red!75!black,very thick}
	]
		\node[font=\bfseries] at (1.5,3.0) {(a) Connected};

		\node[high]   (c0) at (1.5,2.15) {$a$};
		\node[vertex] (c1) at (0.8,1.25) {};
		\node[vertex] (c2) at (1.5,0.95) {};
		\node[vertex] (c3) at (2.2,1.25) {};

		\draw[thick]
			(c0)--(c1)
			(c0)--(c2)
			(c0)--(c3)
			(c1)--(c2);

		\node[
			copy,
			fit=(c0)(c1)(c2)(c3),
			inner sep=7pt
		] {};

		\node[font=\small] at (1.5,0.45)
			{one candidate for \(H\)};

		\node[font=\bfseries] at (6.5,3.0) {(b) Disconnected};

		\node[high]   (h1a) at (5.5,2.15) {$a$};
		\node[vertex] (h1b) at (5.8,1.35) {};
        \node[vertex] (h1c) at (5.2,1.35) {};

		\draw[thick] (h1a)--(h1b);
        \draw[thick] (h1a)--(h1c);

		\node[
			copy,
			fit=(h1a)(h1b)(h1c),
			inner sep=7pt
		] {};

		\node[font=\small] at (5.5,0.65)
			{\(H_1\cong P_3\)};

		\node[vertex] (h2a) at (7.7,2.15) {};
		\node[vertex] (h2b) at (7.25,1.25) {};
		\node[vertex] (h2c) at (8.15,1.25) {};

		\draw[thick]
			(h2a)--(h2b)
			(h2b)--(h2c)
			(h2c)--(h2a);

		\node[
			copy,
			fit=(h2a)(h2b)(h2c),
			inner sep=7pt
		] {};

		\node[font=\small] at (7.7,0.65)
			{\(H_2\cong K_3\)};

		\draw[bad]
			(h1a)--node[
				above,
				sloped,
				fill=white,
				inner sep=1.5pt
			]{\small cross-edge}
			(h2a);

	\end{tikzpicture}

	\caption{Why disconnected forbidden graphs require an extra step.
	In panel~(a), the tester only needs to find and check one connected copy.
	In panel~(b), the two dashed sets separately induce \(P_3\) and \(K_3\)
	in \(G\), but the red edge prevents their union from inducing
	\(P_3\cup K_3\). Red vertices have high degree in \(G\); blue vertices
	have low degree. Edges to vertices outside the candidates are omitted.
	The structural lemma provides a way to choose candidates without
	such cross-edges.}

	\label{fig:disconnected-obstacle}
\end{figure}
\FloatBarrier

\paragraph*{Why an exceptional component is needed.}
Fix an integer \(m\ge1\). Let \(H=P_3\cup K_1\), and let \(G\) consist of a star with center
\(a\) and leaves \(u_1,\ldots,u_{2m}\), together with isolated vertices
\(z_1,\ldots,z_m\); see \Cref{fig:star-isolates}. Every induced
\(P_3\) contains \(a\), so \(G\) has at most one vertex-disjoint
\(P_3\). Yet its distance from induced-\(H\)-freeness is exactly
\(m=(n-1)/3\).

To see the lower bound, each set $\{a,u_{2i-1},u_{2i},z_i\} \text{ where } i\in[m]$, 
induces \(H\). These sets intersect only at \(a\), so no unordered
vertex pair lies in two of them. Every induced-\(H\)-free repair must
change an internal pair of each set, and therefore needs at least \(m\)
edits. Conversely, adding the \(m\) edges \(az_i\) turns \(G\) into
a star, which is induced-\(P_3\cup K_1\)-free.

Removing the vertices of any one \(P_3\) leaves an edgeless graph,
so a search for further vertex-disjoint \(P_3\)'s cannot succeed.
But this first copy can still be extended to an induced \(H\): every
\(z_i\) is nonadjacent to it. When \(2m>d\), the center is a
high-degree root, and the isolated vertices form precisely the separated
region used by the exceptional case. Thus this case is necessary even
for forests.

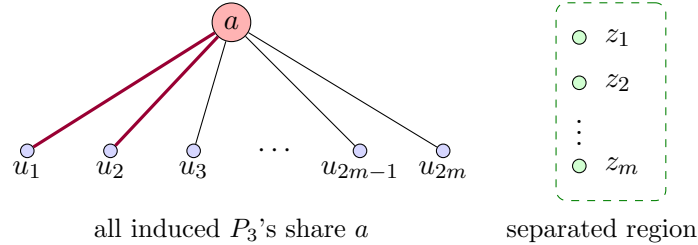
\begin{figure}[!htb]
\centering
\begin{tikzpicture}[
  vertex/.style={circle,draw,fill=blue!16,inner sep=1.8pt},
  root/.style={circle,draw,fill=red!30,inner sep=2.5pt},
  isolated/.style={circle,draw,fill=green!18,inner sep=1.8pt}
]
  \node[root] (a) at (3,2.6) {$a$};
  \foreach \i/\x in {1/0.3,2/1.4,3/2.5,5/4.7,6/5.8}{
    \node[vertex] (u\i) at (\x,0.9) {};
    \draw (a)--(u\i);
  }
  \draw[purple!80!black,very thick] (u1)--(a)--(u2);
  \node[below] at (u1) {$u_1$};
  \node[below] at (u2) {$u_2$};
  \node[below] at (u3) {$u_3$};
  \node at (3.6,0.9) {$\cdots$};
  \node[below] at (u5) {$u_{2m-1}$};
  \node[below] at (u6) {$u_{2m}$};
  \node[font=\small] at (3,-0.15) {all induced \(P_3\)'s share \(a\)};
  \node[isolated] (z1) at (7.6,2.4) {};
  \node[isolated] (z2) at (7.6,1.8) {};
  \node at (7.6,1.25) {$\vdots$};
  \node[isolated] (zm) at (7.6,0.7) {};
  \node[right=3pt of z1] (l1) {$z_1$};
  \node[right=3pt of z2] (l2) {$z_2$};
  \node[right=3pt of zm] (lm) {$z_m$};
  \node[draw=green!50!black,dashed,rounded corners,
    fit=(z1)(z2)(zm)(l1)(l2)(lm),inner sep=6pt] {};
  \node[font=\small] at (7.9,-0.15) {separated region};
\end{tikzpicture}
\caption{A star with \(2m\) leaves and \(m\) isolated vertices.
All induced \(P_3\)'s share the center, but any isolated vertex completes
the highlighted path to an induced \(P_3\cup K_1\). The separated region
allows the tester to use this single path instead of requiring many
vertex-disjoint paths.}
\label{fig:star-isolates}
\end{figure}

\subsection{The main structural lemma}

Throughout the section, write $H=m_0K_1\cup m_1F_1\cup\cdots\cup m_sF_s$, 
where the graphs \(F_1,\ldots,F_s\) are pairwise non-isomorphic, connected,
and nontrivial.  Let \(k\ge2\) be the number of components of \(H\), counted
with multiplicity, and list them as \(H_1,\ldots,H_k\).  Put $I_0:=\{i:H_i\cong K_1\},\ I_+:=[k]\setminus I_0$
and let \(h_i:=|V(H_i)|\), \(h:=|V(H)|\), and
\(h_{\max}:=\max_i h_i\).  Define
\[
	Q:=2k,\qquad
	Q^-:=\begin{cases}1,&k=2,\\2(k-1),&k\ge3,\end{cases}
	\qquad K:=Q(1+h-h_{\max})-1.
\]
The values \(Q\) and \(Q^-\) are the numbers of candidates per component
needed by the selection lemma, for \(k\) and \(k-1\) components,
respectively. The larger bound \(K\) allows us to discard overlapping
candidates before applying that lemma.

Let \(N\) be the constant in \Cref{lem:structural}, and set
\[
d:=\left\lceil\frac{\max\{100N,1000K\}}{\varepsilon^2}\right\rceil,\qquad
G':=\mathbf{Decompose}_{\mathbf R}\left(G,d,\frac{\varepsilon}{10}\right).
\] 
We work with a fixed seed \(\mathbf R\) for which
\(|E(G)\setminus E(G')|\le\varepsilon n/10\), and use this fixed \(G'\)
throughout the section.  Let $V^h:=\{v:\deg_G(v)>d\},\ V^\ell:=V(G)\setminus V^h$, 
and define
\[
	\eta:=\frac{\varepsilon}{1000Kd},\qquad
	\lambda:=\frac{\varepsilon}{100K},\qquad
	\alpha_0:=\frac{\varepsilon}{8(4K+1)},\qquad
	\alpha:=\min\{\alpha_0,\eta\}.
\]
Here \(\alpha\) bounds the edits allowed after removing a witness,
\(\eta\) is the distance needed for searches in a separated region, and
\(\lambda n\) is the region size needed when \(H\) has isolated vertices.
The larger value \(\alpha_0\) is used in the uniqueness proof.
We use the following lower bound on \(n\) so that the costs of
the bounded witness sets fit within the edit budget:
\[
	n_0:=\left\lceil\max\left\{
	\frac{8Kdh}{\varepsilon},
	\frac{8K^2}{\varepsilon},
	\frac{100d(d+1)Kh}{\varepsilon},
	\frac{Q^-(h+1)dh_{\max}}{\eta},
	\frac{2Q^-h}{\lambda}
	\right\}\right\rceil.
\]

The next definition describes the obstruction to repeatedly finding
vertex-disjoint copies. It separates two operations: removing a few copies,
then changing a few edges in the remaining graph.

\begin{definition}[Small removal witness]
	Let \(F\) be connected, let \(X\) be a graph on a subset of \(V(G)\),
	and let \(\beta\in(0,1)\). A \emph{small removal witness} for \(F\) in \(X\),
	with edit budget \(\beta n\), is a family \(\mathcal C\) of at most \(K\)
	pairwise vertex-disjoint induced \(F\)-copies such that, for
	\(U:=\bigcup_{C\in\mathcal C}V(C)\), the graph \(X-U\) can be made
	induced-\(F\)-free by at most \(\beta n\) edge additions or deletions.
	We also call \(\mathcal C\) a \emph{witness family} and \(U\) its
	\emph{witness set}. The empty family is allowed.
\end{definition}

For example, a witness consisting of one induced \(P_3\) says that removing
its three vertices leaves a graph within \(\beta n\) edits of
induced-\(P_3\)-free. It does not say that this was the only induced \(P_3\),
or that removing it already destroys every induced \(P_3\).
The vertex removal here only defines the remaining graph; it is not an
edge-edit operation charged to the budget \(\beta n\).
As in \Cref{pre}, \(n\) always denotes the original number of vertices,
and the bound \(K\) stays fixed throughout this section.

If \(F\) has no small removal witness in \(G'\) with edit budget \(\beta n\), then excluding
the vertices of any family of at most \(K\) disjoint copies leaves a graph
that is still \(\beta\)-far from induced-\(F\)-free. For nontrivial \(F\),
this is exactly the assumption needed to keep applying
\Cref{cor:sequential-induced-finder}. If a witness does exist, the finder
may still find further copies; the definition only says that its
far-input guarantee need not apply after this particular family is removed.

We use the following notation for such a witness.  Suppose that \(H_p\) has a
small removal witness and let \(\mathcal C_p=\{C_1,\ldots,C_t\}\), \(t\le K\), be a
witness family with witness set \(U_p\).  Define
\[
	A:=U_p\cap V^h,
	\qquad
	Z^\ell:=(U_p\cap V^\ell)\cup
	\bigl(\Gamma_G(U_p\cap V^\ell)\cap V^\ell\bigr).
\]
Let \(J\) contain the indices \(j\) for which \(C_j\) has a high-degree
vertex.  Such a vertex is unique by \Cref{lem:structural}; call it \(a_j\).
For each \(j\in J\), define $B_j:=V^\ell\setminus\bigl(\Gamma_G(a_j)\cup Z^\ell\bigr)$. 
The set \(Z^\ell\) contains the low-degree vertices of the witness family
and their low-degree neighbors. To form \(B_j\), we also exclude all
low-degree neighbors of the root \(a_j\). Thus \(B_j\) is disjoint from
\(C_j\), and no edge of the original graph \(G\) joins them. We call \(B_j\)
the \emph{separated region} for \(C_j\); see \Cref{fig:far2-regions}.
The lemma below shows that at least one such region also has enough
copies of every remaining component.

We can now state the main result of the section.  Its proof occupies the rest
of the section.

\begin{samepage}
\begin{lemma}[Two structural alternatives for far inputs]
	\label{far}
	Assume that \(G\) is \(\varepsilon\)-far from induced-\(H\)-free and
	\(n\ge n_0\).  Exactly one of the following alternatives holds.
	\begin{itemize}
		\item[\textnormal{(I)}] No component \(H_i\) has a small removal
		witness in \(G'\) with edit budget \(\alpha n\).
		\item[\textnormal{(II)}] There is exactly one index \(p\) for which
		\(H_p\) has a small removal witness in \(G'\) with edit budget \(\alpha n\).  It is nontrivial, and its isomorphism type occurs only once
		in \(H\).  For every witness family \(\mathcal C_p\), the set \(J\) is
		nonempty and contains an index \(j^\star\) such that
		\(G'[B_{j^\star}]\) is \(\eta\)-far from induced-\(H_i\)-free for
		every \(i\in I_+\setminus\{p\}\).  If \(m_0>0\), then also
		\(|B_{j^\star}|\ge\lambda n\).  In all cases, $E_G\bigl(V(C_{j^\star}),B_{j^\star}\bigr)=\emptyset$.
	\end{itemize}
\end{lemma}
\end{samepage}

	\begin{figure}[t]
		\centering
		
		\begin{tikzpicture}[
			scale=0.9,
			low/.style={circle,draw,fill=blue!18,inner sep=1.7pt},
			high/.style={circle,draw,fill=red!35,inner sep=2.2pt},
			ext/.style={circle,draw,fill=orange!20,inner sep=1.7pt},
			edge/.style={draw,thick},
			witness/.style={draw=purple!70,dashed,rounded corners,thick},
			zregion/.style={draw=orange!80!black,dashed,rounded corners,thick}
			]
			
			\node[font=\bfseries] at (3.7,5.15) {(a) The set $Z^\ell$};
			
			\node[high] (a1) at (1.4,3.15) {$a_1$};
			\node[low]  (u11) at (0.8,2.35) {};
			\node[low]  (u12) at (1.5,2.05) {};
			\node[low]  (u13) at (2.15,2.4) {};
			
			\draw[edge] (a1)--(u11);
			\draw[edge] (a1)--(u12);
			\draw[edge] (a1)--(u13);
			
			\node[witness,
			fit=(a1)(u11)(u12)(u13),
			inner sep=7pt,
			label=above left:$C_1$] {};
			
			\node[ext] (x11) at (0.35,1.15) {};
			\node[ext] (x12) at (1.5,1.05) {};
			\node[ext] (x13) at (2.65,1.15) {};
			
			\draw[edge] (u11)--(x11);
			\draw[edge] (u12)--(x12);
			\draw[edge] (u13)--(x13);
			
			\node[high] (a2) at (4.55,3.15) {$a_2$};
			\node[low]  (u21) at (3.9,2.35) {};
			\node[low]  (u22) at (4.55,2.05) {};
			\node[low]  (u23) at (5.2,2.4) {};
			
			\draw[edge] (a2)--(u21);
			\draw[edge] (a2)--(u22);
			\draw[edge] (a2)--(u23);
			
			\node[witness,
			fit=(a2)(u21)(u22)(u23),
			inner sep=7pt,
			label=above right:$C_2$] {};
			
			\node[ext] (x21) at (3.35,1.15) {};
			\node[ext] (x22) at (4.55,1.05) {};
			\node[ext] (x23) at (5.75,1.15) {};
			
			\draw[edge] (u21)--(x21);
			\draw[edge] (u22)--(x22);
			\draw[edge] (u23)--(x23);
			
			\node[zregion,
			fit=(u11)(u12)(u13)(x11)(x12)(x13)
			(u21)(u22)(u23)(x21)(x22)(x23),
			inner sep=10pt,
			label=below:$Z^\ell$] {};
			
			\node[draw=purple!70,dashed,rounded corners,thick,
			fit=(a1)(u11)(u12)(u13)(a2)(u21)(u22)(u23),
			inner sep=13pt,
			label=above:$U_p$] {};
			
			\node[low,label=right:{\footnotesize $U_p\cap V^\ell$}]
			at (6.75,3.0) {};
			\node[ext,label=right:{\footnotesize
				$\Gamma_G(U_p\cap V^\ell)\cap V^\ell$}]
			at (6.75,2.25) {};
			\node[high,label=right:{\footnotesize $U_p\cap V^h$}]
			at (6.75,1.5) {};
			
		\end{tikzpicture}
		
		\vspace{1.2em}
		
		\begin{tikzpicture}[
			scale=0.9,
			low/.style={circle,draw,fill=blue!18,inner sep=1.7pt},
			high/.style={circle,draw,fill=red!35,inner sep=2.2pt},
			excluded/.style={circle,draw,fill=red!12,inner sep=1.7pt},
			good/.style={circle,draw,fill=green!18,inner sep=1.7pt},
			edge/.style={draw,thick},
			witness/.style={draw=purple!70,dashed,rounded corners,thick},
			rootnbhd/.style={draw=red!70,dashed,rounded corners,thick},
			bregion/.style={draw=green!55!black,dashed,rounded corners,thick},
			zregion/.style={draw=orange!80!black,dashed,rounded corners,thick,
				fill=orange!5}
			]
			
			\node[font=\bfseries] at (4.2,5.2) {(b) The set $B_j$};
			
			\node[high] (aj) at (1.45,3.05) {$a_j$};
			\node[low]  (c1) at (0.75,2.25) {};
			\node[low]  (c2) at (1.45,1.9) {};
			\node[low]  (c3) at (2.15,2.25) {};
			
			\draw[edge] (aj)--(c1);
			\draw[edge] (aj)--(c2);
			\draw[edge] (aj)--(c3);
			\draw[edge] (c1)--(c2);
			
			\node[witness,
			fit=(aj)(c1)(c2)(c3),
			inner sep=8pt,
			label=above left:$C_j$] {};
			
			\node[excluded] (n1) at (2.85,3.6) {};
			\node[excluded] (n2) at (3.3,2.95) {};
			\node[excluded] (n3) at (2.9,1.8) {};
			
			\draw[edge] (aj)--(n1);
			\draw[edge] (aj)--(n2);
			\draw[edge] (aj)--(n3);
			\draw[edge] (c3)--(n3);
			
			\node[rootnbhd,
			fit=(n1)(n2)(n3),
			inner sep=8pt,
			label=above:{\small other neighbors of $a_j$}] {};
			
			\begin{pgfonlayer}{background}
			\node[zregion,fit=(c1)(c2)(c3)(n3),inner sep=7pt,
			label=below:$Z^\ell$] {};
			\end{pgfonlayer}
			
			\node[good] (b1) at (5.8,3.3) {};
			\node[good] (b2) at (6.65,3.0) {};
			\node[good] (b3) at (7.5,3.35) {};
			\node[good] (b4) at (5.95,2.0) {};
			\node[good] (b5) at (6.85,1.75) {};
			\node[good] (b6) at (7.75,2.15) {};
			
			\draw[edge] (b1)--(b2);
			\draw[edge] (b2)--(b3);
			\draw[edge] (b4)--(b5);
			\draw[edge] (b5)--(b6);
			
			\node[bregion,
			fit=(b1)(b2)(b3)(b4)(b5)(b6),
			inner sep=14pt,
			label=above:$B_j$] {};
			
			\node at (4.2,0.4) {$E_G(V(C_j),B_j)=\emptyset$};
			
		\end{tikzpicture}
		
		\caption{
			Schematic illustration of the sets in \Cref{far2}.
			In panel~(a), $Z^\ell
			=
			(U_p\cap V^\ell)
			\cup
			\bigl(\Gamma_G(U_p\cap V^\ell)\cap V^\ell\bigr)$. Thus, \(Z^\ell\) contains the low-degree vertices of the witness family
			and their low-degree neighbors.
			In panel~(b), for a witness copy \(C_j\) with unique high-degree vertex
			\(a_j\), $B_j
			=
			V^\ell\setminus
			\bigl(\Gamma_G(a_j)\cup Z^\ell\bigr)$. By construction, no vertex
			of \(B_j\) is adjacent in \(G\) to a vertex of \(C_j\); thus
			\(E_G(V(C_j),B_j)=\emptyset\).
		}
		\label{fig:far2-regions}
	\end{figure}
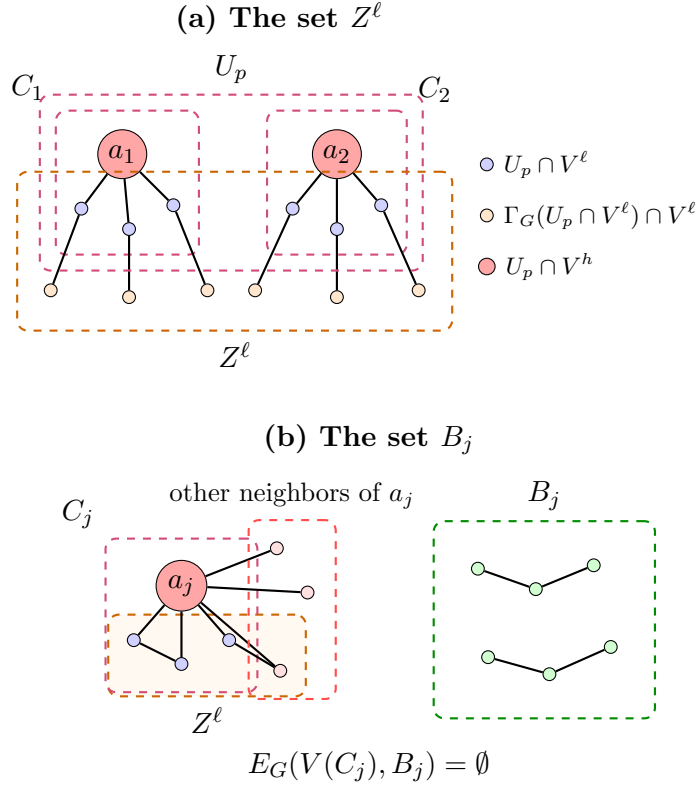

In the star-plus-isolates example, a single \(P_3\) is a removal witness,
and its separated region consists of the isolated vertices. The second
alternative extends this situation to general \(H\): the same region must
contain enough candidates for \emph{every} remaining component.
The guarantee holds for every witness family, which matters because the
tester uses the family returned by its random search.

We now establish the two tools used by the disconnected tester.
The next subsection proves the selection lemma: enough candidates for
each component guarantee a choice with no cross-edges.
We then prove \Cref{far} in two steps. \Cref{far1} shows that at most one
component has a small removal witness, and \Cref{far2} finds a suitable
separated region for every witness family of that component.

\FloatBarrier
\subsection{Rainbow independent transversals}
	We use a special case of the rainbow independent-transversal problem of
	Aharoni, Berger, and Ziv~\cite{ABZ07}.  Call an auxiliary graph \(R\) a
	\emph{\((p,q)\)-candidate graph} if its vertices are split into \(p\) color
	classes \(V_1,\ldots,V_p\) of size \(q\), all
	edges go between different classes, and the graph is outerplanar. A
	\emph{rainbow independent transversal} chooses one vertex from each class
	and contains no edge.  In our application, a vertex represents a candidate
	component copy, and an edge records a forbidden cross-edge between two
	candidates. Once the candidate copies are globally vertex-disjoint,
	contracting each connected copy and deleting all other vertices gives
	an outerplanar graph. Suppressing parallel edges and deleting edges
	within each color class gives the required candidate graph.

	The linear edge bound for outerplanar graphs now gives a short proof:
	with enough candidates, a uniformly random choice from each class has
	positive probability of containing no edge.
	
	\begin{theorem}[Rainbow independent-transversal bound]
		\label{upperbound}
		If \(p\ge2\) and \(q\ge2p\), every
		\((p,q)\)-candidate graph has a rainbow independent transversal.
	\end{theorem}
	
	\begin{proof}
		Choose one vertex independently and uniformly from each color class.
		Each edge has both endpoints selected with probability \(1/q^2\).
		Since \(R\) is outerplanar on \(pq\) vertices, it has at most
		\(2pq-3\) edges. A union bound gives
		\[
		\Pr[\text{some edge has both endpoints selected}]
		\le\frac{|E(R)|}{q^2}
		\le\frac{2pq-3}{q^2}<1,
		\]
		where the last inequality uses \(q\ge2p\). Thus an independent
		choice exists.
	\end{proof}

	For one color, a single candidate suffices. The tester only needs the
	existence of a transversal: it need not query every cross-edge or
	construct the transversal before rejecting.

\subsection{Proof of the structural lemma}

We first show that two different components \(H_p,H_q\) cannot both
have small removal witnesses. Each witness gives a nearby graph with no
induced copy of the corresponding component outside its witness set.
We combine these two repairs and add edges between a few carefully chosen
vertices. These added edges prevent the two components from appearing
as nonadjacent parts of an induced copy of \(H\).

\begin{claim}[At most one component has a small removal witness]
	\label{far1}
	If \(G\) is \(\varepsilon\)-far from induced-\(H\)-free and \(n\ge n_0\),
	then \(K_1\) has no small removal witness in \(G'\) with edit budget \(\alpha_0 n\). At most one nontrivial type \(F_r\) has such a witness,
	and its multiplicity in \(H\) is \(m_r=1\).
\end{claim}
	
	\begin{claimproof}
	The bounds in the definition of \(n_0\) give
	\(Kdh\le\varepsilon n/8\) and \(K^2\le\varepsilon n/8\).
	The claim for $K_1$ now follows directly.
		Indeed, $K^2\le\varepsilon n/8$ and \(\varepsilon<1\) imply
		\(n>K\). After deleting the vertices of at most \(K\) copies of
		$K_1$, at least one vertex remains, but a nonempty graph cannot be
		induced-$K_1$-free, regardless of how its edges are modified.
		
		\emph{Assume that two small removal witnesses exist.}
		It remains to prove the uniqueness and multiplicity statement for the
		nontrivial types. Suppose, for contradiction, that either two distinct
		types \(F_r,F_{r'}\) have small removal witnesses, or a type \(F_r\)
		with such a witness has \(m_r\ge2\). In either case, there are distinct
		indices \(p,q\in[k]\) such that both \(H_p\) and \(H_q\) have small
		removal witnesses in \(G'\) with edit budget \(\alpha_0 n\). Both are isomorphic to one of the $F_r$, and hence
		are connected and contain at least one edge.
		
		For each \(i\in\{p,q\}\), let $\mathcal C_i
		=\{C^i_1,\ldots,C^i_{t_i}\}$, where
		$t_i\le K$ be a witness collection of pairwise vertex-disjoint induced
		\(H_i\)-copies, and put
		$U_i := \bigcup_{C\in\mathcal C_i}V(C)$.
		By the definition of a removal witness, there is an
		induced-\(H_i\)-free graph \(\widehat G_i\) on \(V(G')\setminus U_i\).
		Let $M_i:=E(G'-U_i)\triangle E(\widehat G_i)$. 
		We choose \(\widehat G_i\) so that \(|M_i|\le\alpha_0 n\).

		Set $W^\ell
		:=
		(U_p\cup U_q)\cap V^\ell,\,
		A_p:=U_p\cap V^h$ and $
		A_q:=U_q\cap V^h$. Every copy in \(\mathcal C_p\cup\mathcal C_q\) is connected and
		lies in a single connected component of \(G'\).
		By \Cref{lem:structural}, every connected component of \(G'\)
		contains at most one vertex of \(V^h\). Therefore,
		$|A_p|\le K$ and $|A_q|\le K$ both hold.
		
		Let \(D_0\) be the set of edges of \(G'\) incident with at least
		one vertex of \(W^\ell\). Since every vertex of \(V^\ell\) has
		degree at most \(d\) in \(G\), and \(G'\subseteq G\), $|D_0|
		\le d|W^\ell|
		\le d(|U_p|+|U_q|)
		\le 2Kd|H|
		\le \frac{\varepsilon n}{4}$.
		
		\emph{Build a nearby graph.}
		We construct a graph \(X\) on \(V(G)\). First isolate every
		vertex of \(W^\ell\). Outside \(U_p\cup W^\ell\), use the
		adjacency relation of \(\widehat G_p\); all remaining adjacencies
		are inherited from \(G'\). Formally,
		\begin{enumerate}
			\item every vertex of \(W^\ell\) is isolated in \(X\);
			\item if \(x,y\notin U_p\cup W^\ell\), then $xy\in E(X)
			\Longleftrightarrow
			xy\in E(\widehat G_p)$;
			\item for all remaining pairs, \(X\) retains the adjacency
			relation of \(G'\).
		\end{enumerate}
		Thus, $X[V(G)\setminus(U_p\cup W^\ell)]
		=
		\widehat G_p[V(G)\setminus(U_p\cup W^\ell)]$ and $|E(G')\triangle E(X)|
		\le
		|M_p|+|D_0|
		\le
		\alpha_0 n+\frac{\varepsilon n}{4}$.
		
		Let \(T\) be the set of vertices in
		\(V(G)\setminus W^\ell\) that occur as an endpoint of a pair in
		\(M_p\cup M_q\). Then $|T|
		\le
		2(|M_p|+|M_q|)
		\le
		4\alpha_0 n$.
		
		Finally, obtain \(Y\) from \(X\) by adding all missing edges
		\begin{enumerate}
			\item inside \(A_p\);
			\item between \(A_p\) and \(A_q\); and
			\item between \(A_p\) and \(T\).
		\end{enumerate}
		Thus, \(A_p\) is a clique in \(Y\), and every vertex of \(A_p\)
		is adjacent to every distinct vertex of \(A_q\cup T\).
		The role of these edges is as follows. Any induced \(H_p\)-copy
		in \(Y\) must use \(A_p\). An induced \(H_q\)-copy must use
		\(A_p\cup A_q\cup T\). The added edges then prevent these two copies
		from being different components of an induced \(H\)-copy.
		We verify these two assertions next.
		
		\emph{Why the new graph is induced-\(H\)-free.}
		We claim that \(Y\) is induced-\(H\)-free.
		
		Suppose otherwise, and let \(Q\) be an induced copy of \(H\) in
		\(Y\). Fix an isomorphism
		$\varphi:H\longrightarrow Y[V(Q)],$
		and for each \(i\in[k]\) let
		$Q_i := Y[\varphi(V(H_i))].$
		Consequently, \(Q_1,\ldots,Q_k\) are the indexed connected
		components of \(Q\). In particular, \(Q_p\) and \(Q_q\) are
		distinct components, even if \(H_p\cong H_q\).
		
		We first show that \(Q_p\) contains a vertex of \(A_p\).
		Suppose that \(V(Q_p)\cap U_p=\emptyset\). Since every vertex of
		\(W^\ell\) is isolated in \(Y\), while \(H_p\) is connected and
		nontrivial, \(Q_p\) cannot contain a vertex of \(W^\ell\).
		Hence,
		$V(Q_p)\cap(U_p\cup W^\ell)=\emptyset.$
		On this vertex set, \(X\) agrees with \(\widehat G_p\), and every
		edge added in passing from \(X\) to \(Y\) has an endpoint in
		\(A_p\subseteq U_p\). Thus,
		$Y[V(Q_p)] = \widehat G_p[V(Q_p)],$
		contradicting that \(\widehat G_p\) is induced-\(H_p\)-free.
		Consequently, \(Q_p\) meets \(U_p\). Since every vertex of
		\(U_p\cap V^\ell\) lies in \(W^\ell\) and is isolated in \(Y\),
		\(Q_p\) must contain a vertex $a\in A_p$.
		
		We next show that \(Q_q\) contains a vertex of
		\(A_p\cup A_q\cup T\). If \(Q_q\) meets \(U_q\), then it cannot meet
		\(U_q\cap V^\ell\subseteq W^\ell\), and hence it contains a
		vertex of \(A_q\).
		
		Suppose instead that
		$V(Q_q)\cap U_q=\emptyset.$
		Again, \(Q_q\) cannot contain a vertex of \(W^\ell\). If, in
		addition,
		$V(Q_q)\cap(A_p\cup A_q\cup T)=\emptyset,$
		then \(Q_q\) is also disjoint from \(U_p\), because
		$U_p\subseteq A_p\cup W^\ell.$
		Also, no pair in \(M_p\cup M_q\) has an endpoint in
		\(V(Q_q)\), by the definition of \(T\). Therefore, on
		\(V(Q_q)\), the graphs \(G'\), \(\widehat G_p\), and
		\(\widehat G_q\) all have the same adjacency relation. No edge
		added in passing from \(X\) to \(Y\) has an endpoint in
		\(V(Q_q)\), either. Consequently, $Y[V(Q_q)]
		=
		G'[V(Q_q)]
		=
		\widehat G_q[V(Q_q)]$, contradicting that \(\widehat G_q\) is induced-\(H_q\)-free.
		
		Thus, in every case there exists $b\in V(Q_q)\cap(A_p\cup A_q\cup T)$.
		
		Since \(Q_p\) and \(Q_q\) are distinct components of \(Q\), we have
		\(a\neq b\). But by construction of \(Y\), \(ab\in E(Y)\):
		if \(b\in A_p\), this follows because \(A_p\) is a clique; if
		\(b\in A_q\), because \(A_p\) is complete to \(A_q\); and if
		\(b\in T\), because \(A_p\) is complete to \(T\). This gives an edge
		between two distinct components of \(Q\), a contradiction.
		
		\emph{Edit count.}
		The construction of \(X\) changes at most
		\(\alpha_0 n+\varepsilon n/4\) pairs, while
		\(|T|\le4\alpha_0 n\). Consequently, the number of edges added when passing
		from \(X\) to \(Y\) is at most $\binom{|A_p|}{2}
		+|A_p||A_q|
		+|A_p||T|
		\le
		2K^2+4K\alpha_0 n$. Therefore,
		\[
		\begin{aligned}
			|E(G)\triangle E(Y)|
			&\le
			\frac{\varepsilon n}{10}
			+\alpha_0 n
			+\frac{\varepsilon n}{4}
			+2K^2
			+4K\alpha_0 n\\
			&\le
			\left(
			\frac{1}{10}
			+\frac14
			+\frac14
			+\frac18
			\right)\varepsilon n
			<\varepsilon n.
		\end{aligned}
		\]
		Here we used
		\(K^2\le\varepsilon n/8\) and
		\(\alpha_0=\varepsilon/[8(4K+1)]\).
		
		We have constructed an induced-\(H\)-free graph \(Y\)
		within fewer than \(\varepsilon n\) edge modifications of \(G\),
		contradicting the assumption that \(G\) is
		\(\varepsilon\)-far from being induced-\(H\)-free.
		
		This contradiction rules out both alternatives. Therefore, at most one type \(F_r\) has a small removal witness,
		and it has multiplicity \(m_r=1\). Together with the first paragraph, this proves the claim.
	\end{claimproof}

	We next handle the component \(H_p\) that has a small removal witness.
	We must show that one witness copy can be combined with copies of all
	remaining components. For each rooted witness copy \(C_j\), the region
	\(B_j\) already has no edge to \(C_j\). The issue is whether it contains
	enough candidates. If every \(B_j\) failed this requirement, we could
	combine the resulting local repairs to make \(G\) induced-\(H\)-free
	with fewer than \(\varepsilon n\) edits, a contradiction.
	
	\begin{claim}[Separated region around an exceptional witness]
		\label{far2}
		Suppose that \(G\) is \(\varepsilon\)-far from induced-\(H\)-free,
		\(n\ge n_0\), and \(H_p\) has a small removal witness in \(G'\)
		with edit budget \(\alpha n\).  Assume
		that \(H_p\) is nontrivial and its type occurs only once in \(H\).  For
		any witness family \(\mathcal C_p\), the set \(J\) is nonempty and
		contains an index \(j^\star\) such that \(G'[B_{j^\star}]\) is
		\(\eta\)-far from induced-\(H_i\)-free for every
		\(i\in I_+\setminus\{p\}\).  If \(m_0>0\), then
		\(|B_{j^\star}|\ge\lambda n\).  Moreover, $E_G\bigl(V(C_j),B_j\bigr)=\emptyset
        \text{ for every }j\in J$.
	\end{claim}
	
	\begin{claimproof}
		The witness family gives an induced-\(H_p\)-free graph
		\(\widehat G_p\) on \(V(G)\setminus U_p\) for which
		\[
			M_p:=E(G'-U_p)\triangle E(\widehat G_p),
			\qquad |M_p|\le\alpha n.
		\]
		The numerical bounds needed below now follow from the standing setup.
		Indeed, every witness copy lies in one component of \(G'\), and hence it has at
		most one high-degree vertex.  Hence \(|A|\le K\), and $|Z^\ell|\le(d+1)|U_p\cap V^\ell|
		\le(d+1)Kh$. 
		Because \(n\ge n_0\), this gives
		\(d|Z^\ell|\le\varepsilon n/100\).  Outerplanarity gives
		\(|V^h|<4n/d\), and our choice of \(d\) gives
		\(|A||V^h|\le\varepsilon n/100\).
		
		\emph{A root exists.}
		We first prove that \(J\neq\emptyset\).
		
		Suppose otherwise. Then
		$U_p\subseteq V^\ell.$
		In particular,
		$U_p\subseteq Z^\ell.$
		Starting from \(G\), first pass to \(G'\), using at most
		$\frac{\varepsilon n}{10}$
		edge modifications. Next delete all edges of \(G'\) incident with
		\(U_p\). Since every vertex of \(U_p\) is low-degree, $\left|
		\{e\in E(G'):e\cap U_p\neq\emptyset\}
		\right|
		\le
		d|U_p|
		\le
		d|Z^\ell|
		\le
		\frac{\varepsilon n}{100}$. Finally, on \(V(G)\setminus U_p\), perform
		the at most \(\alpha n\) edits in \(M_p\).
		
		The vertices of \(U_p\) are now isolated, while the graph induced by
		\(V(G)\setminus U_p\) is \(\widehat G_p\), which is
		induced-\(H_p\)-free. Since \(H_p\) is connected and nontrivial, no
		induced copy of \(H_p\) can contain an isolated vertex. Consequently, the
		resulting graph is induced-\(H_p\)-free, and hence it is also
		induced-\(H\)-free. The total number of modifications is at most $\frac{\varepsilon n}{10}
		+
		\frac{\varepsilon n}{100}
		+
		\alpha n
		<
		\varepsilon n$, a contradiction. Therefore,
		$J\neq\emptyset.$
		
		\emph{A common far region exists.}
		We must find one region that is far from induced-\(H_i\)-free
		for every remaining nontrivial component \(H_i\), and that has at least
		\(\lambda n\) vertices if isolated components are needed.
		
		Suppose, for contradiction, that no \(j\in J\) has all the
		required properties.  Put $J_{\mathrm{small}}
		:=
		\{j\in J:m_0>0\text{ and }|B_j|<\lambda n\}$ and $
		J_{\mathrm{large}}:=J\setminus J_{\mathrm{small}}$. If \(I_+\setminus\{p\}=\emptyset\), then every index of
		\(J_{\mathrm{large}}\) would already satisfy the desired conclusion
		(there is no nontrivial component left to check, and the size
		requirement holds), contrary
		to our assumption. Therefore, \(J_{\mathrm{large}}=\emptyset\) in that
		degenerate case. Otherwise, the following choice is available.
		For every \(j\in J_{\mathrm{large}}\), failure of the desired
		conclusion supplies an index
		$i(j)\in I_+\setminus\{p\}$
		such that $G'[B_j]$ is not \(\eta\)-far from being
		induced-\(H_{i(j)}\)-free.
		
		For each \(j\in J_{\mathrm{large}}\), choose an
		induced-\(H_{i(j)}\)-free graph \(\widehat X_j\) on \(B_j\), and let
\[
N_j:=E(G'[B_j])\triangle E(\widehat X_j),\qquad |N_j|\le\eta n.
\]
		Notice that \(B_j\cap U_p=\emptyset\). Indeed, \(B_j\) is a subset
		of \(V^\ell\), and every vertex of \(U_p\cap V^\ell\) belongs to
		\(Z^\ell\).
		Therefore, \(\widehat G_p[B_j]\) is well-defined.

		The regions \(B_j\) may overlap, so their proposed repairs
		\(\widehat X_j\) may disagree about the same pair of vertices.
		We avoid this problem by comparing every repair with the single
		graph \(\widehat G_p\). Let \(T\) contain all endpoints of pairs
		where any comparison disagrees. Outside \(T\), all the required
		repairs agree with \(\widehat G_p\).
		Formally, for \(j\in J_{\mathrm{large}}\), put $L_j
		:=
		E(\widehat G_p[B_j])
		\triangle
		E(\widehat X_j)$. Since $E(\widehat G_p[B_j])
		\triangle
		E(G'[B_j])
		\subseteq M_p$, we have
		$L_j \subseteq M_p\cup N_j.$
		Consequently, $\bigcup_{j\in J_{\mathrm{large}}}L_j
		\subseteq
		M_p\cup
		\bigcup_{j\in J_{\mathrm{large}}}N_j$.
		
		Let \(T\) be the set of all endpoints of pairs belonging to
		$\bigcup_{j\in J_{\mathrm{large}}}L_j.$
		Since every \(B_j\subseteq V^\ell\),
		$T\subseteq V^\ell.$
		Also, since \(|J_{\mathrm{large}}|\le t\le K\), we have $$|T|
		\le
		2\left|
		\bigcup_{j\in J_{\mathrm{large}}}L_j
		\right|
		\le
		2\left(
		|M_p|
		+
		\sum_{j\in J_{\mathrm{large}}}|N_j|
		\right)
		\le
		2(\alpha+K\eta)n.$$
		
		\emph{Build a graph that satisfies all the repairs.}
		We construct \(Y\) so that any induced \(H_p\)-copy must contain a
		root \(a_j\). If \(j\in J_{\mathrm{large}}\), the component
		corresponding to \(H_{i(j)}\) in an induced \(H\)-copy would have
		to lie in \(B_j\setminus T\). This is impossible because
		\(Y[B_j\setminus T]\) agrees with
		\(\widehat X_j[B_j\setminus T]\), which has no induced
		\(H_{i(j)}\)-copy. For a small region, making its root adjacent
		to all other vertices rules out an induced disconnected copy
		containing that root. Here is the construction.
		
		First, on \(V(G)\setminus U_p\), start with the graph
		\(\widehat G_p\), and initially make every vertex of \(U_p\)
		isolated.
		
		Next, isolate every vertex of
		$Z^\ell\cup T$.
		
		For every $a\in A$ and $
		x\in
		V^\ell\setminus(Z^\ell\cup T)$, set $ax\in E(Y)
		\Longleftrightarrow
		ax\in E(G)$.
		
		Finally, add every missing edge inside \(A\) and every missing edge
		between $A$ and $V^h\setminus A$. Thus, \(A\) is a clique in \(Y\), and \(A\) is complete to \(V^h\setminus A\).
		
		For every $j\in J_{\mathrm{small}}$, add all still-missing edges
		from $a_j$ to $V(G)\setminus\{a_j\}$.  Consequently, every small root
		\(a_j\) is \emph{universal} in \(Y\), meaning adjacent to every
		other vertex.  The vertices of $Z^\ell\cup T$
		may now have neighbors in $A$, but they have no neighbors outside
		$A$.
		
		\emph{Why the repair works.}
		We claim that \(Y\) is induced-\(H\)-free.
		
		Suppose otherwise, and let \(Q\) be an induced copy of \(H\) in
		\(Y\).  Since the isomorphism type of $H_p$ has multiplicity one,
		$Q$ contains a unique component
		$Q_p\cong H_p.$
		Fix an isomorphism $\varphi:H\to Y[V(Q)]$ that maps $H_p$ to
		$Q_p$, and for every $i\in[k]$ let
		$Q_i:=Y[\varphi(V(H_i))].$
		These components of \(Q\) are indexed by the components of \(H\);
		they are distinct even when some \(H_i\) are isomorphic.
		
		We first show that \(Q_p\) must intersect \(A\).
		
		Suppose that
		$V(Q_p)\cap U_p=\emptyset.$
		Every vertex in \(Z^\ell\cup T\) has no neighbor outside $A$.
		Since $Q_p$ is disjoint from $U_p\supseteq A$ and $H_p$ is
		connected and nontrivial, such a vertex would be isolated inside
		$Q_p$. Thus,
		$V(Q_p)\cap(Z^\ell\cup T)=\emptyset.$
		On
		$V(G)\setminus (U_p\cup Z^\ell\cup T),$
		all changes made after starting with \(\widehat G_p\) are incident
		with vertices of \(A\subseteq U_p\). Thus,
		$Y[V(Q_p)] = \widehat G_p[V(Q_p)],$
		contradicting the fact that \(\widehat G_p\) is
		induced-\(H_p\)-free.
		
		Therefore, \(Q_p\) intersects \(U_p\). Every vertex of
		$U_p\cap V^\ell$
		belongs to \(Z^\ell\) and has no neighbor outside $A$. Consequently,
		since \(H_p\) is connected and nontrivial, \(Q_p\) must contain a vertex
		$a_j\in A$
		for some \(j\in J\).
		
		If $j\in J_{\mathrm{small}}$, then $a_j$ is universal in $Y$.
		As $H$ is disconnected, $Q$ has a component other than $Q_p$;
		every vertex of that component is adjacent to $a_j$, a
		contradiction.  Therefore,
		$j\in J_{\mathrm{large}}.$
		
		Now consider the component
		$Q_{i(j)}\cong H_{i(j)}$
		of the induced copy \(Q\).
		
		If \(Q_{i(j)}\) contains a high-degree vertex
		$b\in V^h,$
		then
		$a_jb\in E(Y),$
		because \(a_j\in A\) and \(A\) is complete to \(V^h\).
		This is impossible, since \(Q_p\) and \(Q_{i(j)}\) are distinct
		components of the induced copy \(Q\). Thus,
		$V(Q_{i(j)})\subseteq V^\ell.$
		
		Every vertex of $Z^\ell\cup T$ has no neighbor outside $A$.
		The component $Q_{i(j)}$ contains no vertex of $A$, by the
		preceding high-degree argument.  Consequently, such a vertex would be
		isolated inside $Q_{i(j)}$. Since $H_{i(j)}$ is connected and
		nontrivial, $V(Q_{i(j)})
		\cap
		(Z^\ell\cup T)
		=\emptyset$.
		
		Also, $V(Q_{i(j)})
		\cap
		\Gamma_G(a_j)
		=
		\emptyset$. Indeed, if \(x\in V(Q_{i(j)})\cap\Gamma_G(a_j)\), then
		\(x\in V^\ell\setminus(Z^\ell\cup T)\), and by the construction of
		\(Y\),
		$a_jx\in E(Y),$
		again contradicting the fact that \(Q_p\) and \(Q_{i(j)}\) are
		distinct components of an induced copy. Consequently, $V(Q_{i(j)})
		\subseteq
		V^\ell
		\setminus
		\bigl(
		\Gamma_G(a_j)\cup Z^\ell\cup T
		\bigr)
		=
		B_j\setminus T$.
		
		By the definition of \(T\), no pair in \(L_j\) has an endpoint in
		\(B_j\setminus T\). Thus, $\widehat G_p[B_j\setminus T]
		=
		\widehat X_j[B_j\setminus T]$. Also, on pairs of low-degree vertices outside
		\(Z^\ell\cup T\), the graph \(Y\) agrees with \(\widehat G_p\).
		Consequently, $Y[V(Q_{i(j)})]
		=
		\widehat X_j[V(Q_{i(j)})]$. This contradicts the fact that \(\widehat X_j\) is
		induced-\(H_{i(j)}\)-free. Therefore, \(Y\) is induced-\(H\)-free.
		
		\emph{Edit count.}
		We compare the final adjacency of every pair directly with its
		adjacency in \(G\). In particular, temporarily removing and then
		restoring an edge does not contribute to the edit distance.
		Pairs outside $U_p\cup Z^\ell\cup T$ are charged to the
		decomposition or to $M_p$.  Deleting the original edges incident
		with $Z^\ell\cup T$ costs at most
		$d(|Z^\ell|+|T|)$.
		Before the last edge-addition step, pairs
		between $A$ and the remaining low-degree vertices agree with
		$G$, while completing $A$ to the high-degree vertices costs at
		most $|A||V^h|$.  For a small root, every still-missing
		low-degree adjacency outside $Z^\ell\cup T$ has its low endpoint in
		$B_j$.
		Thus, making all small roots adjacent to every other vertex adds at most
		$K\bigl(\lambda n+|Z^\ell|+|T|\bigr)$
		further edges. Consequently,
		\[
		\begin{aligned}
			|E(G)\triangle E(Y)|
			&\le
			|E(G)\triangle E(G')|
			+|M_p|
			+d|Z^\ell|
			+d|T|
			+|A||V^h|
			+K\bigl(\lambda n+|Z^\ell|+|T|\bigr)\\
			&\le
			\frac{\varepsilon n}{10}
			+\alpha n
			+\frac{\varepsilon n}{100}
			+2d(\alpha+K\eta)n
			+\frac{\varepsilon n}{100}
			+K\lambda n
			+K|Z^\ell|
			+K|T|.
		\end{aligned}
		\]
		
		Since $\alpha
		\le
		\frac{\varepsilon}{1000Kd}$ and $
		\eta
		=
		\frac{\varepsilon}{1000Kd}$, we obtain $$\alpha n
		\le
		\frac{\varepsilon n}{1000Kd},\quad
		2d\alpha n
		\le
		\frac{2\varepsilon n}{1000K}\text{\quad and\quad}
		2dK\eta n
		=
		\frac{2\varepsilon n}{1000}.$$
		Also, since $d\ge400K/\varepsilon>K$,
		$$K|Z^\ell|\le d|Z^\ell|\le\frac{\varepsilon n}{100},\qquad
		K|T|\le d|T|\le2d(\alpha+K\eta)n\text{\quad and\quad}K\lambda n=\varepsilon n/100.$$ Substituting these estimates
		into the preceding display gives
		\[
		|E(G)\triangle E(Y)|
		\le
		\left(
		\frac{14}{100}
		+\frac{1}{1000Kd}
		+\frac{4}{1000K}
		+\frac{4}{1000}
		\right)\varepsilon n
		<\frac{3\varepsilon n}{20}
		<\varepsilon n.
		\]
		This contradicts the assumption that \(G\) is
		\(\varepsilon\)-far from being induced-\(H\)-free.
		
		Therefore, there exists some \(j^\star\in J\) such that, simultaneously
		for every \(i\in I_+\setminus\{p\}\),
		$G'[B_{j^\star}]$
		is \(\eta\)-far from being induced-\(H_i\)-free.
		If $m_0>0$, the same index also satisfies
		$|B_{j^\star}|\ge\lambda n$.
		
		\emph{The region is separated.}
		Fix any \(j\in J\).
		Let \(x\in B_j\). By definition,
		$x\notin\Gamma_G(a_j),$
		and hence \(x\) has no edge in \(G\) to the unique high-degree vertex
		\(a_j\) of \(C_j\).
		
		Now let
		$u\in V(C_j)\cap V^\ell.$
		Since
		$u\in U_p\cap V^\ell,$
		every low-degree neighbor of \(u\) belongs to $\Gamma_G(U_p\cap V^\ell)\cap V^\ell
		\subseteq Z^\ell$.
		But
		$B_j\cap Z^\ell=\emptyset.$
		Since \(x\in B_j\subseteq V^\ell\),
		$ux\notin E(G).$
		Therefore, no vertex of \(B_j\) is adjacent in \(G\) to any vertex of
		\(C_j\), and hence
		$E_G(V(C_j),B_j)=\emptyset.$
		Since \(j\in J\) was arbitrary, this holds for every rooted witness
		copy.
	\end{claimproof}

	We can now finish the proof of the dichotomy.
	
	\begin{proof}[Proof of \Cref{far}]
		Because \(\alpha\le\alpha_0\), every small removal witness with
		edit budget \(\alpha n\) is also one with edit budget \(\alpha_0 n\). By \Cref{far1},
		at most one component of \(H\) has such a witness; it is nontrivial
		and its type occurs only once in \(H\).

		If no component has a small removal witness with edit budget \(\alpha n\),
		alternative~\textnormal{(I)} holds. Otherwise, let \(H_p\) be the
		unique component with such a witness and fix any witness family.  All hypotheses of
		\Cref{far2} hold, and hence that claim gives a rooted copy
		\(C_{j^\star}\), its separated region \(B_{j^\star}\), and every
		property in alternative~\textnormal{(II)}.  The alternatives are disjoint,
		which proves the lemma.
	\end{proof}

	\section{The disconnected case: algorithm and analysis}
	\label{sec:disconnected-tester}
	
	This section turns the structural dichotomy from \Cref{far} into a tester for
	a disconnected forbidden graph. Its main result is the following theorem.
	
	\begin{theorem}[Disconnected tester with threshold-list access]
		\label{thm:test-disconnected}
		Let $H=m_0K_1\cup m_1F_1\cup\cdots\cup m_sF_s$ be a fixed disconnected
		graph, where $F_1,\ldots,F_s$ are pairwise non-isomorphic connected
		nontrivial graphs. Then there is a threshold
		\(d=(1/\varepsilon)^{O_H(1)}\) for which induced-\(H\)-freeness in
		outerplanar graphs is testable with one-sided error under threshold-list access. Its total query complexity is
		\((1/\varepsilon)^{O_H(1)}\).
	\end{theorem}
	
	The algorithm first runs the connected finder repeatedly for each
	component of \(H\). If every search finds enough copies, the selection
	lemma guarantees mutually nonadjacent copies whose union induces \(H\).
	If just one search stops early, the algorithm tries each rooted copy
	found by that search and looks for the remaining components in its
	separated region. On a far input, \Cref{far} guarantees that one of these
	regions works whenever the search subroutines meet their success
	guarantees. In both branches, the exact checks certify the individual
	copies, and the selection lemma guarantees that a suitable union exists.
	This gives one-sided error.
	
	The proof follows the same order as the connected-case proof. We first fix the
	parameters and the graph queries. We then prove the local search
	tools used in the exceptional region, state the tester, prove its correctness,
	and bound its queries. Together with \Cref{thm:test-connected}, this proves
	\Cref{thm:sparse-main}. The next section then simulates the adjacency-list
	tester in the random-neighbor model.
	
	\subsection{Setting of parameters}
	
	Keep the component notation \(H_i,I_0,I_+,h_i,h,h_{\max}\) and the
	parameters \(Q,Q^-,K,d,\eta,\lambda,\alpha\) from
	\Cref{sec:disconnected-structure}. In particular, the same choice of \(d\)
	allows us to apply \Cref{lem:structural} with parameter
	\(\varepsilon/10\), and satisfies \(d\ge400K/\varepsilon\).
	Set \(n_{\mathrm{far}}:=n_0\), the size threshold in that section, and
	let
	\[
		n_{\mathrm{alg}}
		:=
		\max\left\{
		n_{\mathrm{far}},
		\left\lceil
		\frac{2Q^-h\,dh_{\max}}{\eta}
		\right\rceil
		\right\}.
	\]
	Also let
	\(L:=k(K+1)+K(k-1)Q^-\) and
	\(\delta_\star:=1/(20L)\).
		
		As in the connected case, we first work with exact threshold-list queries and
	direct random-neighbor samples at high-degree roots. Every call to
	\Cref{Local} uses the same seed as the current decomposition. The proof of \Cref{thm:sparse-main} later implements these operations
	exactly in the
	adjacency-list model.
	
	\subsection{Local tools for the exceptional region}

	\paragraph*{Local access to the separated region.}
	The exceptional case of \Cref{far2} provides a separated region \(B_j\),
	but the tester does not have this region explicitly. To use it
	algorithmically, we first show that, once the witness family
	\(\mathcal C_p\) is known, membership in \(B_j\) can be decided locally
	with a number of queries that is independent of \(n\).
	
	For a known family \(\mathcal C_p\) of at most \(K\) induced
	\(H_p\)-copies in \(G'\), use \(U_p,Z^\ell,a_j,B_j\) as defined in
	\Cref{sec:disconnected-structure}. Testing membership in these regions
	does not require knowing whether \(\mathcal C_p\) is a removal witness.

	\begin{lemma}[Testing membership in a separated region]
		\label{lem:Bj-local}
		The set \(Z^\ell\) can be constructed using \(O(Khd)\)
		threshold-list queries. Once it is stored, membership of a vertex
		in any \(B_j\) can be tested with one more threshold-list query.
	\end{lemma}

	\begin{proof}
		Since
		$|U_p|\le K|H|,$
		call $\mathsf{TL}_d$ on every vertex of $U_p$ and on every neighbor
		in a returned low-degree list. The exact lists and classifications determine \(Z^\ell\) using
		$O(Khd)$
		threshold-list queries.
		
		To test whether a vertex \(v\) belongs to \(B_j\), call
		$\mathsf{TL}_d(v)$.  A $\mathsf{High}$ answer excludes $v$.
		Otherwise check whether \(v\) lies in the stored set \(Z^\ell\)
		and whether \(a_j\) lies in the returned list \(\Gamma_G(v)\).
		The vertex belongs to \(B_j\) exactly when both answers are no.
		Looking up \(v\) in the stored set needs no graph query.
	\end{proof}
	
	When \(H\) has isolated components, \Cref{far2} additionally guarantees
	that the relevant separated region contains a linear number of vertices.
	Since membership in \(B_j\) is locally decidable by
	\Cref{lem:Bj-local}, uniform vertex sampling then suffices to find
	singleton candidates in \(B_j\), even after excluding a bounded set of
	previously selected vertices.
	
	\begin{lemma}[Finding a vertex in a large region]
		\label{lem:find-singleton-region}
		Let \(B\subseteq V(G)\), and suppose that membership in \(B\)
		can be tested using at most \(q_B\) threshold-list queries. Suppose that
		$|B|\ge\lambda n.$
		Let \(S\subseteq V(G)\) be a known set of at most \(M\) forbidden
		vertices. If
		$n\ge\frac{2M}{\lambda},$
		then, for every \(\delta\in(0,1)\), a vertex in \(B\setminus S\) can
		be found with probability at least \(1-\delta\), using $N_{\mathrm{sing}}(\lambda,\delta)
		=
		O\left(
		\frac{q_B+1}{\lambda}\log\frac2\delta
		\right)$ threshold-list and direct random-neighbor queries.
	\end{lemma}
	
	\begin{proof}
		A uniformly random vertex belongs to \(B\setminus S\) with probability
		at least $\frac{|B|-|S|}{n}
		\ge
		\lambda-\frac Mn
		\ge
		\frac{\lambda}{2}$.
		Sample vertices independently and uniformly with replacement, test
		membership in \(B\), and return the first sampled vertex outside \(S\).
		After $\left\lceil
		\frac{2}{\lambda}\log\frac1\delta
		\right\rceil$ samples, the failure probability is at most \(\delta\).  Membership in
		the explicit set \(S\) uses no graph query.  For the regions \(B_j\) of
		\Cref{lem:Bj-local}, one may take \(q_B=O(1)\) threshold-list queries after
		$Z^\ell$ has been constructed.
	\end{proof}

	\paragraph*{Finding an induced copy inside \(B_j\).}
	The preceding lemmas allow us to recognize the separated region \(B_j\)
	locally and, when isolated components are present, to sample singleton
	candidates from it. For a nontrivial component, \Cref{far2}
	provides a stronger structural guarantee: the induced subgraph on the
	relevant region remains far from being induced-\(F\)-free. We
	need a version of the connected finder that searches entirely inside a
	low-degree region whose membership can be tested, while avoiding
	previously selected vertices. The argument is simple: farness gives many
	vertex-disjoint copies in the region, and excluding a bounded set of
	vertices removes only a few of them. Sampling a vertex in one of the
	remaining copies and exploring its neighborhood then finds the copy.

	\begin{lemma}[Low-degree induced-copy finder]
		\label{lem:find-low-region}
		Let \(F\) be a fixed connected nontrivial graph and put
		\(r:=|V(F)|\).
		Let \(B\subseteq V^\ell\) be a set whose membership can be decided using
		a constant number of threshold-list queries, and suppose that
		\(X:=G'[B]\) is \(\eta\)-far from being induced-\(F\)-free, where the
		distance is normalized by \(n=|V(G)|\).
		
		Let \(S\subseteq B\) be a known set of at most \(M\) forbidden
		vertices. If $n\ge\frac{2drM}{\eta}$,
		then, for every \(\delta\in(0,1)\), an induced copy of \(F\)
		contained in \(X-S\) can be found with probability at least
		\(1-\delta\), using $N_{\mathrm{low}}(F,\eta,\delta)
		=
		O\left(
		\frac d\eta
		\log\frac2\delta
		\left(
		d^r(N_{\mathrm{sim}}+1)
		+
		d^{r^2}
		\right)
		\right)$
		queries, where \(N_{\mathrm{sim}}\) is the query bound from
		\Cref{simulate}. As with \(N_{\mathrm{sim}}\), the notation
		\(N_{\mathrm{low}}\) suppresses the decomposition parameters
		\(d\) and \(\varepsilon\).
	\end{lemma}

	\begin{proof}
		Let \(\mathcal P\) be a maximal family of pairwise
		vertex-disjoint induced copies of \(F\) in \(X\).
		
		We first claim that
		$|\mathcal P| \ge \frac{\eta n}{dr}.$
		Otherwise, putting
		$W:=\bigcup_{C\in\mathcal P}V(C),$
		isolating \(W\) requires at most \(d|W|=dr|\mathcal P|<\eta n\)
		edge deletions.
		
		Every vertex of \(W\) becomes isolated. By maximality of
		\(\mathcal P\), there is no induced \(F\)-copy in \(X-W\).
		Since \(F\) is connected and nontrivial, the resulting graph is
		induced-\(F\)-free, contradicting the \(\eta\)-farness of \(X\).
		
		Since the members of \(\mathcal P\) are pairwise vertex-disjoint,
		at most \(M\) of them intersect \(S\). The assumption
		$n\ge\frac{2drM}{\eta}$
		implies
		$|\mathcal P|\ge2M.$
		Therefore, at least half of the members of \(\mathcal P\) avoid \(S\).
		Let \(\mathcal P'\) be this subfamily. Then $\left|
		\bigcup_{C\in\mathcal P'}V(C)
		\right|
		=
		r|\mathcal P'|
		\ge
		\frac{\eta n}{2d}$. Consequently, a uniformly random vertex of \(G\) belongs to one of these
		copies with probability at least
		$\frac{\eta}{2d}.$
		
		Whenever this event occurs, explore \(G'\) from the sampled vertex
		to depth \(r\), expanding only vertices belonging to
		\(B\setminus S\). Since every vertex of \(B\) has degree at most
		\(d\), the entire copy containing the sampled vertex is exposed.
		Enumerate the constant-size candidate copies and verify each
		candidate exactly in \(G\) using exact threshold-list answers.
		
		Therefore, $O\left(
		\frac d\eta\log\frac2\delta
		\right)$ independent repetitions suffice. Each repetition uses $$O_F\left(
		d^r(N_{\mathrm{sim}}+1)
		+
		d^{r^2}
		\right)$$ threshold-list and direct random-neighbor queries, proving the claimed bound.
	\end{proof}

	\subsection{The disconnected tester}
	We now combine the preceding structural results and local finders into
	a tester for a disconnected graph $H$.
	
	The tester follows the dichotomy established in
	\Cref{far1,far2}. It first searches sequentially for \(K+1\) pairwise
	vertex-disjoint copies of every nontrivial component, using the
	sequential finder from \Cref{cor:sequential-induced-finder}; isolated
	components are handled directly. If all component searches
	succeed, the choice of \(K\), together with the rainbow independent-transversal
	bound of \Cref{upperbound}, allows us to select one copy of each component so that the selected copies are mutually disjoint and have no
	cross-edges in \(G\), and hence they form an induced copy of \(H\).
	
	The only remaining case on a far input is the exceptional situation
	identified by \Cref{far1}: a unique nontrivial component \(H_p\), whose
	isomorphism type has multiplicity one, has a small removal witness
	in \(G'\) with edit budget \(\alpha n\).
	For a witness family
	\(\mathcal C_p=\{C_1,\ldots,C_t\}\), \Cref{far2} guarantees a rooted witness
	copy \(C_{j^\star}\) and a separated low-degree region \(B_{j^\star}\) in
	which every remaining nontrivial component is still far from induced-freeness;
	when isolated components are present, \(B_{j^\star}\) also contains a linear
	number of vertices. Also, $E_G(V(C_{j^\star}),B_{j^\star})=\emptyset$.
	The tester examines each rooted witness copy in turn. By
	\Cref{lem:Bj-local}, membership in the corresponding region \(B_j\) can be
	tested locally, while
	\Cref{lem:find-low-region,lem:find-singleton-region} find the required
	nontrivial and isolated components inside that region. A second application
	of \Cref{upperbound}, now to the remaining \(k-1\) components,
	then completes an induced copy of \(H\).
	
	The tester only has to accept or reject; it does not have to output
	an induced copy. Thus, it need not find the rainbow independent
	transversal explicitly before rejecting. Once all candidate families have been found,
	outerplanarity and \Cref{upperbound} guarantee that the input graph contains
	the required transversal. Thus rejection is safe even though some cross-edges
	between high-degree candidates cannot be queried directly.
	
	The role of \Cref{cor:sequential-induced-finder} is now explicit.
	If a component has no small removal witness, its search finds \(K+1\)
	copies with high probability. A search that stops early does not by itself
	prove that a witness exists: the finder may simply have failed.
	But on the event that every call on a far graph succeeds, stopping
	means that the remaining graph is not \(\alpha\)-far. The copies already
	found then form the small removal witness needed in alternative~(II).

	In the algorithm, \textsc{Full} means that \(K+1\) copies have been
	found. It does not mean that all copies have been found.
	\textsc{Stopped} means that a finder call returned no copy.

	\begin{algorithm}[h]
		\SetAlgoLined
		\caption{\textbf{TestInducedDisconnected} (one execution)}
		\label{alg:test-disconnected}
		\KwIn{an outerplanar graph \(G\), a fixed disconnected graph
			\(H=H_1\cup\cdots\cup H_k\), and \(\varepsilon\)}
		\KwOut{\textnormal{\textsc{Accept}} or \textnormal{\textsc{Reject}}}
		
		\If{$n<n_{\mathrm{alg}}$}{
			Call $\mathsf{TL}_{n-1}(v)$ for every $v\in[n]$, reconstruct
			$G$, and return the exact answer\;
		}
		
		Choose a fresh seed $\mathbf R$, use it throughout this execution,
		and set $G':=\mathbf{Decompose}_{\mathbf R}(G,d,\varepsilon/10)$\;
		
		\ForEach{$i\in I_0$}{
			Let $\mathcal C_i$ consist of any $K+1$ distinct singleton
			vertices and mark $i$ \textnormal{\textsc{Full}}\;
		}
		\ForEach{$i\in I_+$}{
			Set $\mathcal C_i:=\emptyset$\;
			Run at most $K+1$ stages. At each stage, put
			$S_i:=\bigcup_{C\in\mathcal C_i}V(C)$ and run the connected
			induced-copy finder from \Cref{thm:find-one-induced}, with excluded
			set \(S_i\), target
			$H_i$, farness parameter $\alpha$, and failure probability
			$\delta_\star$\;
			If the finder returns a copy, add it to $\mathcal C_i$. If it
			returns no copy, mark $i$ \textnormal{\textsc{Stopped}} and stop
			the search for this $i$\;
			If $|\mathcal C_i|=K+1$, mark $i$ \textnormal{\textsc{Full}}\;
		}
		
		\If{every index is \textnormal{\textsc{Full}}}{
			\Return{\textnormal{\textsc{Reject}}}
		}
		\If{there is not exactly one \textnormal{\textsc{Stopped}} index}{
			\Return{\textnormal{\textsc{Accept}}}
		}
		
		Let $p$ be the unique \textnormal{\textsc{Stopped}} index and write
		$\mathcal C_p=\{C_1,\ldots,C_t\}$, where $t\le K$\;
		\Return{\textbf{ExceptionalSearch}\((G,G',p,\mathcal C_p)\)}
	\end{algorithm}
	
	\begin{algorithm}[H]
		\SetAlgoLined
		\caption{\textbf{ExceptionalSearch}}
		\label{alg:test-disconnected-exceptional}
		\KwIn{\(G\), its fixed-seed decomposition \(G'\), the stopped index
			\(p\), and \(\mathcal C_p=\{C_1,\ldots,C_t\}\)}
		\KwOut{\textnormal{\textsc{Accept}} or \textnormal{\textsc{Reject}}}
		
		Put $U_p:=\bigcup_{j=1}^tV(C_j)$,
		$Z^\ell:=(U_p\cap V^\ell)\cup
		(\Gamma_G(U_p\cap V^\ell)\cap V^\ell)$, and
		$J:=\{j\in[t]:|V(C_j)\cap V^h|=1\}$\;
		
		\If{$J=\emptyset$}{
			\Return{\textnormal{\textsc{Accept}}}
		}
		
		\ForEach{$j\in J$}{
			Let $a_j$ be the unique high-degree vertex of $C_j$ and put
			$B_j:=V^\ell\setminus(\Gamma_G(a_j)\cup Z^\ell)$\;
			Set $S:=\emptyset$\;
			Process the component indices $i\ne p$ in any fixed order, seeking
			$Q^-$ candidates of each color. For $i\in I_+$, invoke
			\Cref{lem:find-low-region} sequentially on $G'[B_j]-S$ with
			target $H_i$; for $i\in I_0$, invoke
			\Cref{lem:find-singleton-region} sequentially on
			$B_j\setminus S$. Give every call failure probability
			$\delta_\star$, and add the vertices of every returned candidate
			to $S$\;
			If any invocation fails, abandon this $j$ and continue with the
			next choice in $J$\;
			\If{all $(k-1)Q^-$ invocations succeed}{
				\Return{\textnormal{\textsc{Reject}}}
			}
		}
		
		\Return{\textnormal{\textsc{Accept}}}
	\end{algorithm}
	\paragraph*{Implementation details.}
	If \(I_+=\emptyset\), then $H=m_0K_1$, and hence the
	initialization phase marks every component index \textnormal{\textsc{Full}}.
	For \(n\ge n_{\mathrm{alg}}\), the rainbow argument below then certifies
	an induced copy of \(H\); the small-input branch is handled exactly.
	The amplified structural tester runs two independent complete executions of
	\Cref{alg:test-disconnected}, including its call to
	\Cref{alg:test-disconnected-exceptional}, with fresh decomposition seeds and
	fresh finder randomness, and rejects if either execution rejects.
	
	\subsection{Correctness analysis}
	We first prove that the tester never rejects an induced-\(H\)-free graph. We
	then condition on the decomposition and finder success events and use the two
	alternatives of \Cref{far} to prove rejection on a far graph.
	
	\begin{proof}[Correctness of \Cref{thm:test-disconnected}]
		We first prove completeness. Suppose that \(G\) is
		induced-\(H\)-free.
		
		Assume first that the procedure in \Cref{alg:test-disconnected} rejects
		because every component is marked
		\textnormal{\textsc{Full}}. Therefore, for
		every \(i\in[k]\), the algorithm has found
		$K+1 = Q(1+h-h_{\max})$
		pairwise vertex-disjoint induced copies of \(H_i\).
		
		Relabel the components so that \(h_k=h_{\max}\). Greedily select
		\(Q\) candidates of each component, choosing \(H_k\) last.
		Before selecting the candidates for a given component, at most
		$Q(h-h_{\max})$
		vertices have been used. Since the available copies of the current
		component are pairwise vertex-disjoint, at most that many of them
		intersect the previously used vertices. Consequently, at least \(Q\)
		candidates remain. We obtain \(Q\) candidates of every
		component so that all chosen candidates are pairwise vertex-disjoint.
		
		Construct a \(k\)-colored auxiliary graph whose vertices are these
		candidate copies, with two vertices of different colors adjacent
		exactly when there is an edge of \(G\) between the corresponding
		copies. This auxiliary graph is obtained from a subgraph of \(G\)
		by contractions and edge deletions. It is outerplanar. Every color class has
		$Q=2k$
		vertices. By \Cref{upperbound}, the auxiliary graph has a rainbow
		independent transversal. The corresponding copies have no
		cross-edge in \(G\), and hence their union is an induced copy of
		\(H\) in \(G\), a contradiction.
		
		Now suppose that the algorithm rejects in the exceptional branch,
		for some \(p\) and \(j\in J\). It has found \(Q^-\) globally
		vertex-disjoint induced copies of every \(H_i\), \(i\neq p\),
		inside \(B_j\).
		
		If \(k=2\), then \(Q^-=1\), and hence there is only one remaining
		component. If \(k\ge3\), construct the corresponding
		\((k-1)\)-colored auxiliary graph. It is again outerplanar, and every color class has
		$Q^-=2(k-1)$
		vertices. By \Cref{upperbound}, one can select one copy of every
		\(H_i\), \(i\neq p\), with no cross-edge between distinct selected
		copies.
		
		Also, by the definition of \(B_j\),
		$E_G(V(C_j),B_j)=\emptyset.$
		Indeed, \(B_j\) contains no neighbor of the high-degree vertex
		\(a_j\). If
		$u\in V(C_j)\cap V^\ell,$
		then \(u\in U_p\cap V^\ell\), and every low-degree neighbor of
		\(u\) belongs to \(Z^\ell\), but
		\(B_j\cap Z^\ell=\emptyset\). Thus, \(C_j\), together with the
		selected copies inside \(B_j\), forms an induced copy of \(H\) in
		\(G\), again a contradiction.
		
		Therefore, an induced-\(H\)-free input is always accepted.
		
		We now prove soundness for one execution. Suppose that \(G\) is
		\(\varepsilon\)-far from being induced-\(H\)-free.  Let $\mathcal G_{\mathrm{cut}}
		:=
		\left\{|E(G)\setminus E(G')|\le\varepsilon n/10\right\}$
		be the cut-bound event, and condition on \(\mathcal G_{\mathrm{cut}}\).
		By \Cref{lem:structural},
		$\Pr[\mathcal G_{\mathrm{cut}}]\ge2/3$.
		Since $n\ge n_{\mathrm{alg}}\ge n_{\mathrm{far}}$,
		all numerical hypotheses of \Cref{far} are satisfied.
		
		Let \(\mathcal G_{\mathrm{find}}\) be the event that every invocation of the connected
		finder whose current residual graph is \(\alpha\)-far succeeds, every
		invocation of the low-degree finder whose hypotheses hold succeeds,
		and every singleton-finder invocation whose region has the asserted
		linear size succeeds.  There are at most
		$L = k(K+1)+K(k-1)Q^-$
		relevant invocations. By the residual guarantee in
		\Cref{thm:find-one-induced}, each applicable connected-finder call
		has the claimed success probability conditional on the preceding
		history. The regional finders have the same conditional guarantee
		because they use fresh samples. Thus each applicable call has
		conditional failure probability at most
		$\delta_\star=\frac{1}{20L},$
		a union bound gives
		$\Pr[\mathcal G_{\mathrm{find}}\mid\mathcal G_{\mathrm{cut}}]\ge
		1-L\delta_\star = \frac{19}{20}.$
		We now analyze the execution on
		\(\mathcal G_{\mathrm{cut}}\cap\mathcal G_{\mathrm{find}}\).
		
		By \Cref{far}, there are two cases.
		
		\medskip
		\noindent
		\textbf{Case I.}
		No component \(H_i\) has a small removal witness in \(G'\)
		with edit budget \(\alpha n\).
		
		After the vertices of any at most \(K\) pairwise vertex-disjoint
		induced \(H_i\)-copies have been removed, the residual graph
		remains \(\alpha\)-far from being induced-\(H_i\)-free.
		Thus \Cref{cor:sequential-induced-finder} applies. Under
		$\mathcal G_{\mathrm{find}}$, the sequential search for every
		\(i\in I_+\) finds \(K+1\) copies; every \(i\in I_0\) was marked
		\textnormal{\textsc{Full}} directly. Consequently, every component is
		marked \textnormal{\textsc{Full}}, and the algorithm rejects.
		
		\medskip
		\noindent
		\textbf{Case II.}
		There is a unique \(p\in[k]\) such that \(H_p\) has a small
		removal witness in \(G'\) with edit budget \(\alpha n\).
		
		By \Cref{far1}, \(p\in I_+\), and the isomorphism type of
		\(H_p\) occurs exactly once. For every
		\(i\in I_+\setminus\{p\}\), the graph \(H_i\) has no such witness. Therefore, the corresponding sequential search finds
		\(K+1\) pairwise vertex-disjoint induced copies of \(H_i\). Every
		\(i\in I_0\) is marked \textnormal{\textsc{Full}} directly.
		
		If the search for \(H_p\) also finds \(K+1\) copies, then every
		component is marked \textnormal{\textsc{Full}}, and the
		algorithm rejects. Otherwise, suppose that the search for \(H_p\) stops after finding $\mathcal C_p
		=
		\{C_1,\ldots,C_t\}$ where $t\le K$. Since the search for $H_p$ has stopped while
		$\mathcal G_{\mathrm{find}}$ occurs, the residual graph $G'\left[
		V(G')\setminus
		\bigcup_{j=1}^tV(C_j)
		\right]$
		cannot be \(\alpha\)-far from being induced-\(H_p\)-free: otherwise
		the connected finder would have succeeded. Therefore, \(\mathcal C_p\) is a small removal witness for
		\(H_p\) in \(G'\) with edit budget \(\alpha n\).
		
		\Cref{far2} applies to every such witness family, including this
		one returned by the algorithm. Its size hypothesis holds because
		\(n\ge n_{\mathrm{alg}}\ge n_{\mathrm{far}}=n_0\).

		Consequently,
		$J\neq\emptyset,$
		and there exists some
		$j^\star\in J$
		such that, simultaneously for every
		\(i\in I_+\setminus\{p\}\),
		$G'[B_{j^\star}]$
		is \(\eta\)-far from being induced-\(H_i\)-free. If \(m_0>0\), then
		also
		$|B_{j^\star}|\ge\lambda n.$
		
		The algorithm tries every \(j\in J\), and hence it eventually tries
		\(j^\star\). During this trial, the set \(S\) of already selected
		vertices always satisfies
		$|S|\le Q^-h.$
		For every \(i\in I_+\setminus\{p\}\), putting \(r=h_i\), the
		definition of \(n_{\mathrm{alg}}\) gives $n
		\ge
		\frac{2Q^-h\,dh_{\max}}{\eta}
		\ge
		\frac{2dr|S|}{\eta}$.
		Therefore, \Cref{lem:find-low-region} applies at every nontrivial
		stage. At every singleton stage, the definition of
		\(n_{\mathrm{far}}\) similarly gives $n\ge\frac{2Q^-h}{\lambda}\ge\frac{2|S|}{\lambda}$,
		and hence \Cref{lem:find-singleton-region} applies. The algorithm finds all
		required candidates for every remaining component and rejects.
		
		Consequently, whenever \(\mathcal G_{\mathrm{cut}}\cap\mathcal G_{\mathrm{find}}\) occurs,
		the algorithm rejects.  Therefore, conditioned on
		\(\mathcal G_{\mathrm{cut}}\), an \(\varepsilon\)-far input is
		rejected with probability at least \(19/20\). Thus, one execution rejects with probability at least
		$(2/3)(19/20)=19/30$.  Run two independent complete executions,
		with fresh partition seeds and fresh finder randomness, and reject if
		either rejects.  The rejection probability is at least
		$1-(11/30)^2=779/900>2/3$, while one-sidedness is preserved.
	\end{proof}

	\subsection{Query complexity}
	We now bound each phase of the tester and verify the polynomial dependence on
	\(1/\varepsilon\) claimed in \Cref{thm:test-disconnected}.
	
	Let $r_{\max}
	:=
	\max\bigl(\{h_i:i\in I_+\}\cup\{1\}\bigr)$. Recall from the definition preceding \Cref{alg:test-disconnected} that $L:=k(K+1)+K(k-1)Q^-,\,\delta_\star:=\frac{1}{20L}$.
	
	By \Cref{lem:structural,simulate}, for the decomposition parameter
	\(\zeta=\varepsilon/10\), we have $N_{\mathrm{part}}
	=
	\operatorname{poly}
	\left(
	\frac{d^2}{\varepsilon}
	\right)$ and $
	N_{\mathrm{sim}}
	=
	O(dN_{\mathrm{part}})$.
	
	Using the query bound \(N_{\mathrm{find}}\) from
	\Cref{thm:find-one-induced}, define
	\[
	N_{\mathrm{find}}^{\max}
	:=
	\begin{cases}
		\displaystyle
		\max_{i\in I_+}\{
		N_{\mathrm{find}}(H_i,\alpha,\delta_\star)\},
		& I_+\ne\emptyset,\\[1ex]
		0,
		& I_+=\emptyset.
	\end{cases}
	\]
	If \(I_+\ne\emptyset\), then $N_{\mathrm{find}}^{\max}
	=
	O\left(
	\max_{i\in I_+}\left\{
	\frac{\log(1/\delta_\star)}
	{p_{H_i}(\alpha,d)}
	\left(
	d^{h_i}N_{\mathrm{sim}}
	+
	d^{h_i^2+1}
	\right)\right\}
	\right)$.
	
	Similarly, using the query bound \(N_{\mathrm{low}}\) from
	\Cref{lem:find-low-region}, define
	\[
	N_{\mathrm{low}}^{\max}
	:=
	\begin{cases}
		\displaystyle
		\max_{i\in I_+}
		N_{\mathrm{low}}(H_i,\eta,\delta_\star),
		& I_+\ne\emptyset,\\[1ex]
		0,
		& I_+=\emptyset.
	\end{cases}
	\]
	If \(I_+\ne\emptyset\), then $N_{\mathrm{low}}^{\max}
	=
	O\left(
	\frac d\eta
	\log\frac1{\delta_\star}
	\left(
	d^{r_{\max}}(N_{\mathrm{sim}}+1)
	+
	d^{r_{\max}^2}
	\right)
	\right)$.
	
	By \Cref{lem:find-singleton-region,lem:Bj-local}, let $N_{\mathrm{sing}}
	:=
	O\left(
	\frac1\lambda\log\frac2{\delta_\star}
	\right)$
	bound the number of queries needed for one singleton search in a region
	\(B_j\).
	
	In one execution, the initial sequential phase makes at most
	\(k(K+1)\) calls to the connected induced-copy finder and uses at
	most $k(K+1)N_{\mathrm{find}}^{\max}$
	threshold-list and direct random-neighbor queries.
	
	If the exceptional phase is entered, constructing \(Z^\ell\) costs
	\(O(Khd)\) threshold-list queries. There are at most \(K\) rooted witness
	copies, and for each one the algorithm makes at most \((k-1)Q^-\) calls,
	each either to the low-degree finder or to the singleton finder. Consequently, the
	exceptional phase uses at most $O(Khd)
	+
	K(k-1)Q^-
	\max\{N_{\mathrm{low}}^{\max},N_{\mathrm{sing}}\}$
	threshold-list and direct random-neighbor queries.
	
	Therefore, one execution uses at most $k(K+1)N_{\mathrm{find}}^{\max}
	+
	K(k-1)Q^-
	\max\{N_{\mathrm{low}}^{\max},N_{\mathrm{sing}}\}
	+
	O(Khd)$
	queries. Since the tester performs two independent executions, its total
	query complexity is at most twice this quantity.
	
	All quantities $k,\ h,\ K,\ Q,\ Q^-,
	\ \delta_\star$ depend only on \(H\). Also, $d=\left(\frac1\varepsilon\right)^{O_H(1)},\ 
	\alpha,\eta,\lambda
	\ge
	\varepsilon^{O_H(1)}$ 
	and
	\(N_{\mathrm{part}}=(1/\varepsilon)^{O_H(1)}\).
	The bounds above and the estimate
	\(p_{H_i}(\alpha,d)\ge
	\alpha^{O_H(1)}d^{-O_H(1)}\)
	therefore give total query complexity $\left(\frac1\varepsilon\right)^{O_H(1)}$. 
	If \(n<n_{\mathrm{alg}}\), the structural tester reconstructs the graph with
	\(n\) calls to \(\mathsf{TL}_{n-1}\). These are \(n\) threshold-list calls and \(O(n^2)\) primitive adjacency-list queries. This is still
	\((1/\varepsilon)^{O_H(1)}\) because
	\(n_{\mathrm{alg}}=(1/\varepsilon)^{O_H(1)}\).
	This proves the query bound and completes the proof of
	\Cref{thm:test-disconnected}.

	\paragraph*{Completing the adjacency-list theorem.}
	The connected and disconnected analyses now cover every fixed nonempty
	forbidden graph. It remains to replace each exact threshold-list call by its
	adjacency-list implementation.
	
	\begin{proof}[Proof of \Cref{thm:sparse-main}]
		If \(H\) is not outerplanar, accept: an induced subgraph of an outerplanar
		graph is outerplanar. If \(H=K_1\), the known value of \(n\) gives the
		answer. Otherwise, run the
		tester from \Cref{thm:test-connected} when \(H\) is connected and the tester
		from \Cref{thm:test-disconnected} when \(H\) is disconnected.
		
		If the disconnected tester takes its small-input branch, answer its \(n\)
		calls to \(\mathsf{TL}_{n-1}\) exactly. Each costs at most \(n\)
		adjacency-list queries, for a total of \(O(n^2)\le
		(1/\varepsilon)^{O_H(1)}\).
		
		In every other branch, answer each call to \(\mathsf{TL}_d(v)\) exactly. First query
		\(\deg_G(v)\). If it is greater than \(d\), return \(\mathsf{High}\).
		Otherwise, query all \(\deg_G(v)\le d\) indexed neighbors and return the
		complete list. Therefore, one threshold-list call costs at most \(d+1\)
		adjacency-list queries. For each requested uniform neighbor of a high-degree
		root \(a\), query \(\deg_G(a)\), choose a uniform index in
		\([\deg_G(a)]\), and query that neighbor.
		
		Both structural testers make
		\(\left(\frac1\varepsilon\right)^{O_H(1)}\)
		threshold-list and high-root sampling calls. They always accept an induced-\(H\)-free graph,
		and they reject an \(\varepsilon\)-far graph with probability at least
		\(1-(11/30)^2=779/900>2/3\). Their displayed parameters and query bounds are
		polynomial in \(1/\varepsilon\) for each fixed \(H\). The total number of
		adjacency-list queries is \((1/\varepsilon)^{O_H(1)}\).
	\end{proof}
	
	\section{From adjacency-list access to the random-neighbor model}
	\label{sec:rn-simulation}
	
	This section proves the main random-neighbor result, \Cref{thm:main}, by
	simulating the adjacency-list tester from \Cref{thm:sparse-main}. The main
	issue is that a random-neighbor query does not reveal a degree or a complete
	neighbor list. We solve this by sampling enough times to recover every
	bounded-degree list used in one adaptive execution. Conditioned on one
	simultaneous recovery event, the simulated run is exactly the same as the
	adjacency-list run. The recovery event can fail with small probability, which
	is why the final tester has two-sided error.
	
	We first show that all neighborhood recoveries succeed together
	with high probability, and that correct recovery gives an exact simulation. We then give the
	full random-neighbor algorithm, prove its error bound, and bound its query
	complexity. The one-sided lower bound in \Cref{app:one-sided-barrier} shows
	that the two-sided error is needed in general.
	
	\subsection{Recovering bounded-degree neighborhoods}
	\label{sec:rn-recovery}
	
	The adjacency-list algorithms use exact threshold-list access. This subsection
	prepares the move to the random-neighbor model. Its main results are
	\Cref{lem:rn-cache,lem:rn-faithful}: the first gives one recovery event for
	all bounded-degree lists used in an adaptive run, and the second shows that
	the whole simulated run agrees with the exact run on that event. The adjacency-list tester needs the complete
	neighborhood of a bounded-degree vertex and must recognize when a vertex is
	high-degree. The random-neighbor oracle does not provide either operation
	directly.
	The following subsection applies these results to the tester and proves
	\Cref{thm:main}.
	
	Recall the threshold-list query \(\mathsf{TL}_t\) from
	\Cref{pre}. We now show that a bounded number of these
	queries can be recovered at once, with high probability, using only
	random-neighbor samples.
	
	The recovery procedure uses two elementary observations.  If
	\(\deg_G(v)>t\), then observing \(t+1\) distinct neighbors certifies that
	\(v\) is high-degree.  If \(\deg_G(v)\le t\), repeated random-neighbor
	samples recover the entire neighborhood with high probability by the
	coupon-collector bound.  For each
	pair \((v,t)\), the answer determined on the first call is reused on all later calls; this keeps adaptive queries consistent.
	
	\begin{algorithm}[H]
		\SetAlgoLined
		\caption{\textbf{Recover}$(v,t,M)$}
		\label{alg:recover}
		\KwIn{a vertex $v$, a threshold $t$, and a sample budget $M$}
		\KwOut{$\mathsf{High}$, or $\mathsf{Low}$ and a sorted neighbor list}
		\If{an answer for $(v,t)$ has already been fixed}{
			\Return{that answer.}}
		Set $S:=\emptyset$\;
		\For{$M$ independent repetitions}{
			Query $u:=\mathsf{RN}(v)$\;
			\If{$u\ne\perp$}{
				add $u$ to $S$}
			\If{$|S|=t+1$}{
				fix the answer as $\mathsf{High}$ and \Return{$\mathsf{High}$.}}
		}
		Fix the answer as $(\mathsf{Low},S^{\uparrow})$ and return it.
	\end{algorithm}
	
	\begin{lemma}[Simultaneous neighborhood recovery]
		\label{lem:rn-cache}
		Fix \(D\ge1\), \(B\ge1\), and \(\eta\in(0,1)\), and set $M
		:=
		\left\lceil
		(D+1)\log\frac{(D+1)B}{\eta}
		\right\rceil$. Suppose that during an execution at most \(B\) distinct pairs
		\((v,t)\) are queried through \Cref{alg:recover}, and every threshold
		\(t\) is at most \(D\).  If every call uses \(M\) random-neighbor samples, then,
		with probability at least \(1-\eta\), all recovered answers are
		simultaneously correct.  The total number of random-neighbor queries is at
		most \(BM\).
	\end{lemma}
	
	\begin{proof}
		Consider the first call to \Cref{alg:recover} for a pair \((v,t)\), and condition on the complete history before the \(M\) fresh random-neighbor samples used for this pair. Put $k:=\deg_G(v)$.  If $k=0$, every reply is $\perp$ and the answer is exact.
		If $1\le k\le t$, the only failure is that some neighbor is missed.  Consequently, $\Pr[S\ne \Gamma_G(v)]
		\le k\left(1-\frac1k\right)^M
		\le D e^{-M/D}
		\le \frac{\eta}{B}$.
		
		Now suppose $k\ge t+1$.  The waiting time to see $t+1$ distinct neighbors is
		the sum, for $i=0,\ldots,t$, of independent geometric waiting times with
		success probabilities $(k-i)/k$.  These probabilities increase with $k$.
		Therefore, failure is maximized at $k=t+1$, where it is the event that at least one
		of $t+1$ coupons is missed.  Thus, $\Pr[|S|\le t]
		\le (t+1)\left(1-\frac1{t+1}\right)^M
		\le (D+1)e^{-M/(D+1)}
		\le \frac{\eta}{B}$.
		The same conditional bound holds whenever a pair \((v,t)\) is queried for the first time, regardless of the earlier history and even if an answer fixed earlier is incorrect. Since at most \(B\) distinct pairs \((v,t)\) are queried during the execution, a union bound proves the simultaneous guarantee. Reusing the answer from the first call keeps the answers consistent, and the query bound follows directly.
	\end{proof}
	
	Simultaneous correctness is important because the structural tester is
	adaptive.  In particular, different calls to the partition oracle must be
	answered with respect to the same bounded-degree graph and the same fixed
	random seed. Consequently, it is not enough that each local reconstruction is
	correct with high probability in isolation.  The next lemma shows that,
	whenever all recovered neighborhoods are correct, the entire execution
	coincides with an execution in which the threshold-list queries are answered
	exactly.
	
	\begin{lemma}[Correct recovery gives an exact simulation]
		\label{lem:rn-faithful}
		Consider an algorithm that queries degrees only to distinguish
		\(\deg_G(v)\le d\) from \(\deg_G(v)>d\), reads complete neighbor
		lists only in the first case, and samples uniform neighbors of
		high-degree vertices. In an \emph{exact run}, all these operations
		are answered exactly, with lists sorted by vertex identifier.
		In a \emph{simulated run}, answer each threshold-list query by
		\(\textbf{Recover}(v,d,M)\), reusing the answer fixed on the first
		call, and use fresh direct \(\mathsf{RN}\) queries for high-degree
		neighbor samples.

		Compare the two runs using the same internal random choices,
		including partition-oracle seeds, and the same sampled neighbors.
		If every recovered answer is correct, the two runs make the same
		queries, receive the same answers, and give the same output.
		In particular, all partition-oracle calls refer to one partition,
		and all calls to \Cref{Local} refer to one graph \(G'_{\mathbf R}\).
	\end{lemma}
	
	\begin{proof}
		Assume that the simultaneous correctness event from \Cref{lem:rn-cache}
		occurs.  Then recovery returns exactly the increasing-ID list of every
		low-degree vertex and classifies every high-degree vertex correctly.  A
		neighbor query in
		$\widetilde G=(V(G),E(G[V^\ell]))$
		is answered by taking the recovered list of a low vertex, retaining
		exactly its low neighbors, and preserving increasing-ID order; a high vertex
		has the empty list.  Therefore, every primitive query made by
		$\mathbf{findPartition}$ is answered according to one fixed representation of
		$\widetilde G$.  With the same seed it returns the same part on every call.
		
		Recovered lists then reconstruct exactly the graph on a partition part and
		its high-degree neighbors.  The set of high-degree neighbors, the chosen multiway cut, the
		component after the cut, and its label are therefore the same as
		in the exact run, and hence every
		\Cref{Local} answer refers to the same globally defined
		$G'_{\mathbf R}$.  Use the same random neighbor of \(a\) in both runs.
		Inducting over the sequence of operations now gives the same
		low-degree explorations,
		high-root branches, separated-region membership answers, candidate sets, and
		induced-copy verification.  In the last operation every tested pair has a low
		endpoint whose complete recovered list decides the adjacency exactly.
	\end{proof}

	\subsection{The random-neighbor tester}
	
	We now combine the connected and disconnected adjacency-list algorithms with
	the recovery guarantees from \Cref{sec:rn-recovery}. The main result of this
	subsection is \Cref{thm:main}. The algorithm replaces every threshold-list
	call by neighborhood recovery, stores each answer for later reuse,
	and keeps direct random-neighbor samples at high-degree roots unchanged.
	We first give the algorithm, then compare it with an exact run using the
	same random choices, and finally bound the error probability and the number
	of queries.
	
	For the large-input branch, let \(\mathcal A\) denote the relevant structural
	tester: the tester from \Cref{thm:test-connected} when \(H\) is connected,
	and the tester from \Cref{thm:test-disconnected} when \(H\) is disconnected.
	In either case, \(\mathcal A\) includes the two independent complete
	executions used for amplification.
	
	By the query bounds in \Cref{thm:test-connected} and
	\Cref{thm:test-disconnected}, there exist fixed
	constants $B_0=B_0(H,\varepsilon)\ge 1
	\text{ and }
	R_0=R_0(H,\varepsilon)\ge 0$
	such that every execution of \(\mathcal A\) involves at most \(B_0\)
	distinct threshold-list records and at most \(R_0\) direct
	random-neighbor queries at high-degree roots. We use these bounds in the
	random-neighbor simulation below.

	\begin{algorithm}[h]
		\SetAlgoLined
		\caption{\textbf{TestInducedRN}}
		\label{alg:test-rn}
		\KwIn{an outerplanar graph \(G\), a fixed graph \(H\), and
			\(\varepsilon\)}
		\KwOut{\textnormal{\textsc{Accept}} or \textnormal{\textsc{Reject}}}
		
		\If{$H$ is not outerplanar}{
			\Return{\textnormal{\textsc{Accept}}}
		}
		
		\If{$H=K_1$}{
			\Return{\textnormal{\textsc{Reject}} iff $n>0$}
		}
		
		\If{$H$ is disconnected and $n<n_{\mathrm{alg}}$}{
			Set $D_0:=\max\{1,n_{\mathrm{alg}}-1\},\,
			B_s:=n_{\mathrm{alg}}$
			and $M_s:=
			\left\lceil
			(D_0+1)
			\log\bigl(12(D_0+1)B_s\bigr)
			\right\rceil$\;
			Discard all previously fixed recovery answers and call
			\textbf{Recover}$(v,n-1,M_s)$ for every $v\in[n]$\;
			Construct the simple graph \(\widehat G\) in which \(uv\) is an edge
			exactly when each recovered list contains the other endpoint, and
			return the exact answer for \(\widehat G\)\;
		}
		
		Let \(\mathcal A\) be the amplified structural tester described above,
		and set $M_0:=
		\left\lceil
		(d+1)\log\bigl(12(d+1)B_0\bigr)
		\right\rceil$\;
		Discard all previously fixed recovery answers and simulate $\mathcal A$\;
		Whenever \(\mathcal A\) calls \(\mathsf{TL}_d(v)\), answer it by
		\textbf{Recover}$(v,d,M_0)$; whenever \(\mathcal A\) requests a
		random neighbor of a high-degree root, issue a fresh call to
		\(\mathsf{RN}(v)\)\;
		If the simulation would create more than \(B_0\) distinct recovered
		records or make more than \(R_0\) direct high-root random-neighbor
		calls, return \textnormal{\textsc{Accept}}\;
		\Return{the output of the simulated structural tester.}
	\end{algorithm}

	\begin{proof}[Proof of \Cref{thm:main}]
		We analyze \Cref{alg:test-rn}. The case \(H=K_1\) is decided exactly,
		and every outerplanar input is induced-\(H\)-free when \(H\) is not
		outerplanar. We may therefore assume that \(H\neq K_1\) is outerplanar.
		
		First consider the small-input branch, which can arise only when \(H\)
		is disconnected and \(n<n_{\mathrm{alg}}\). Since every vertex has
		degree at most \(n-1\le D_0\), we may apply \Cref{lem:rn-cache} with
		\(D=D_0\), \(B=B_s\), and \(\eta=1/12\). By the choice of \(M_s\),
		with probability at least \(11/12\), every call
		\textbf{Recover}\((v,n-1,M_s)\) returns the complete neighbor list of
		\(v\). On this event the algorithm reconstructs \(G\) exactly and
		returns the correct answer. Consequently, this branch has error
		probability at most \(1/12\). It uses at most \(nM_s\)
		random-neighbor queries. Since
		\(n_{\mathrm{alg}}=(1/\varepsilon)^{O_H(1)}\), this branch uses
		\(\left(\frac1\varepsilon\right)^{O_H(1)}\) queries.
		
		We now consider the large-input branch. Let \(\mathcal A\) be the
		corresponding amplified structural tester. By
		\Cref{thm:test-connected,thm:test-disconnected}, if \(G\) is
		induced-\(H\)-free, then \(\mathcal A\) accepts surely, while if
		\(G\) is \(\varepsilon\)-far from being induced-\(H\)-free, then
		\(\mathcal A\) rejects with probability at least $1-\left(\frac{11}{30}\right)^2
		=
		\frac{779}{900}$.
		
		Let \(\mathcal F\) be the event that at least one recovered
		threshold-list record created during the simulation is incorrect.
		Before terminating, the simulation creates at most \(B_0\) such
		records. Therefore, by \Cref{lem:rn-cache} with
		\(D=d\), \(B=B_0\), and \(\eta=1/12\), the choice of \(M_0\) gives $\Pr[\mathcal F]\le\frac1{12}$.
		
		Suppose that \(\mathcal F\) does not occur. By
		\Cref{lem:rn-faithful}, the simulated run and the exact run of the
		structural tester agree at every step when they use the same
		internal random choices and direct random-neighbor outcomes. Every such
		execution respects the deterministic bounds \(B_0\) and \(R_0\).
		Therefore, on \(\neg\mathcal F\), the simulation never terminates because
		either bound is exceeded and returns exactly the output of the
		corresponding structural-tester execution.
		
		Consequently, with the two runs compared in this way, the output of
		\Cref{alg:test-rn} can differ from the output of \(\mathcal A\) only
		if \(\mathcal F\) occurs. In particular, if \(G\) is
		induced-\(H\)-free, then
		\[
		\Pr[
		\textnormal{\textsc{Accept}}
		\mid
		G\text{ is induced-}H\text{-free}
		]
		\ge
		1-\frac1{12}
		=
		\frac{11}{12}.
		\]
		If \(G\) is \(\varepsilon\)-far from being induced-\(H\)-free, then
		\[
		\Pr[
		\textnormal{\textsc{Reject}}
		\mid
		G\text{ is }\varepsilon\text{-far}
		]
		\ge
		\frac{779}{900}-\frac1{12}
		=
		\frac{176}{225}
		>
		\frac23.
		\]
		
		Finally, the large-input branch makes at most $B_0M_0+R_0$
		random-neighbor queries. The structural bounds give
		\(B_0,R_0,d=(1/\varepsilon)^{O_H(1)}\), and the definition of \(M_0\)
		adds only a logarithmic factor. Therefore, this quantity is
		\(\left(\frac1\varepsilon\right)^{O_H(1)}\). Together with the
		small-input analysis, this
		proves \Cref{thm:main}.
	\end{proof}

\section*{Acknowledgments}
The work is supported in part by National Natural Science Foundation of China
(Grant No. 62272431 and 62672464). 

		\bibliography{ref}

	\appendix
	\section{The one-sided random-neighbor barrier}
	\label{app:one-sided-barrier}
	
	The main theorem gives a two-sided random-neighbor tester, while the
	adjacency-list tester is one-sided. This appendix shows that the difference
	is needed. It proves \Cref{prop:one-sided-barrier} by comparing disjoint
	unions of \(P_3\)'s with disjoint unions of triangles. The proof shows that
	no constant-query one-sided random-neighbor tester can distinguish these
	families with the required guarantee.
	
	\begin{proof}[Proof of \Cref{prop:one-sided-barrier}]
		Let \(n=3m\). Partition the labels into \(m\) triples. On every triple let
		\(G_m\) induce a copy of \(P_3\), and obtain \(Y_m\) by adding the missing
		edge, so that every triple induces a triangle. Both graphs are outerplanar
		and have maximum degree at most two, while \(Y_m\) is induced-\(P_3\)-free.
		Each triple of \(G_m\) remains an induced \(P_3\) unless at least one of its
		three internal pairs is edited. Consequently, the distance of \(G_m\) from
		induced-\(P_3\)-freeness is exactly \(m=n/3\).
		
		Also, \(E(G_m)\subseteq E(Y_m)\). Fix an integer \(q\ge 0\), and let
		\(\mathcal A\) be any randomized algorithm that makes at most \(q\)
		random-neighbor queries. For every complete record \(\tau\) of its queries, replies,
		and output, we have $\Pr[\tau\mid G_m]
		\le
		2^q\Pr[\tau\mid Y_m]$.
		Indeed, condition on the internal random coins of \(\mathcal A\). At each
		query, the probability of the fixed oracle answer on \(G_m\) is at
		most twice its probability on \(Y_m\), because every queried vertex has
		degree one or two in \(G_m\) and degree two in \(Y_m\). Multiplying these
		conditional bounds over at most \(q\) queries and then averaging over the
		internal coins gives the displayed inequality.
		
		Let \(\mathcal T_{\mathrm{rej}}\) be the set of complete records on
		which \(\mathcal A\) rejects. Summing the displayed inequality over
		\(\tau\in\mathcal T_{\mathrm{rej}}\) gives $\Pr[\mathcal A\text{ rejects }G_m]
		\le
		2^q
		\Pr[\mathcal A\text{ rejects }Y_m]$.
		A one-sided tester has zero rejection probability on the
		induced-\(P_3\)-free graph \(Y_m\), and hence it is also zero on the
		\(1/4\)-far graph \(G_m\), contradicting soundness.
	\end{proof}

	\section{Proof of the decomposition lemma}\label{app:decomposition}
	\Cref{sec:global-partition} stated the decomposition used throughout the
	tester and deferred its edge-count proof. This appendix supplies that
	proof. It bounds the high--high edges, the edges crossing low-degree
	partition parts, and the multiway cuts. It then checks the two structural
	properties of every component and derives the polynomial choice of the
	degree threshold. These bounds complete the proof of
	\Cref{lem:structural}.

	\begin{proof}[Proof of \Cref{lem:structural}]
		For the empty graph, take the empty partition; the conclusions are immediate. Assume $n\ge1$. Let $V^h:=\{v\in V(G):\deg_G(v)>d\}$ and $V^\ell:=V(G)\setminus V^h$. Since $G$ is outerplanar, $|E(G)|\le2n-3$ when $n\ge2$; the trivial graphs also satisfy the degree-sum bound. Hence $\sum_{v\in V(G)}\deg_G(v)<4n$. Therefore, $|V^h|<\frac{4n}{d}$.
		We bound the number of edges removed while constructing $G_1$, $G_2$, and $G_3$ in \Cref{GP}.
		
		First, \(G[V^h]\) is outerplanar, and hence forming $G_1$ removes at most $ |E(G[V^h])|
		\le 2|V^h|
		<\frac{8n}{d}
		\le \frac{\varepsilon n}{4} $, where the last inequality follows from $d\ge\frac{240}{\varepsilon}\log(2s+1)$.
		
		Second, we apply the partition oracle to $\widetilde G$ with proximity
		parameter \(\varepsilon/(4d)\).  Condition on the event guaranteed by
		\Cref{partition} that at most
		$[\varepsilon/(4d)]d|V(\widetilde G)|$ edges cross partition parts.
		Since $|V(\widetilde G)|=n$, the number of edges of \(G[V^\ell]\)
		deleted when passing from $G_1$ to $G_2$ is at most
		$\frac{\varepsilon}{4d}\cdot dn=\frac{\varepsilon n}{4}$.
		
		It remains to bound the multiway cuts removed when passing from $G_2$ to $G_3$. For every partition part $P$, the terminal set from \Cref{GP} satisfies $T_P=\Gamma_{G_2}(P)\cap V^h=\Gamma_G(P)\cap V^h$, because the earlier deletions leave all low--high edges unchanged. Only parts with $|T_P|\ge2$ require a nonempty multiway cut.
		As in the proof of \cite[Theorem 3.5]{BKN16}, by
		\cite[Claim 3.4]{BKN16}, $\sum_{P:\,|T_P|\ge2}|T_P|
		\le 15|V^h|$. Also, by \cite[Corollary 3.3]{BKN16}, for every such part
		\(P\) there exists a \(T_P\)-multiway cut \(M_P\) of size $|M_P|
		\le
		2(|T_P|-1)\log(2|P|+1)
		\le
		2|T_P|\log(2s+1)$. Since \Cref{GP} chooses a minimum-size multiway cut, its chosen cut
		satisfies the same bound. Therefore, if \(M\) is the union of all
		multiway cuts deleted from $G_2$, then
		\[
		\begin{aligned}
			|M|
			&\le
			2\log(2s+1)
			\sum_{P:\,|T_P|\ge2}|T_P|\\
			&\le
			30|V^h|\log(2s+1)\\
			&<
			\frac{120n}{d}\log(2s+1)\le
			\frac{\varepsilon n}{2}.
		\end{aligned}
		\]
		
		Consequently, before the final restoration step, the total number of removed
		edges is at most $\frac{\varepsilon n}{4}
		+
		\frac{\varepsilon n}{4}
		+
		\frac{\varepsilon n}{2}
		=
		\varepsilon n$. Passing from $G_3$ to $G'_{\mathbf R}$ only restores edges of the original graph \(G\), and hence it cannot
		increase the distance from \(G\). Therefore, $G'_{\mathbf R}$ is
		$\varepsilon$-close to $G$.
		
		We next verify the component structure. In $G_1$, there is no edge
		between two vertices of \(V^h\). In $G_2$, no edge joins two
		distinct low-degree partition parts. Finally, for each part \(P\),
		the \(T_P\)-multiway cut separates every pair of vertices in \(T_P\).
		Any path between distinct high-degree vertices would contain a subpath within one part joining distinct terminals, which its multiway cut forbids. Therefore, no connected component of $G_3$ contains two
		vertices of \(V^h\).
		
		In the restoration step, edges are restored only between vertices already belonging
		to the same connected component of $G_3$. Thus, this step does not merge two
		different components, and every connected component of $G'_{\mathbf R}$ still
		contains at most one vertex of \(V^h\), equivalently at most one
		vertex whose degree in \(G\) is at least \(d+1\).
		
		Also, restoration adds every original edge of \(G\) whose endpoints
		lie in the same connected component \(C\). Since no edge outside
		\(E(G)\) is ever added, we obtain $G'_{\mathbf R}\subseteq G$ and $G'_{\mathbf R}[C]=G[C]$ for every connected component \(C\) of $G'_{\mathbf R}$.
		
		Finally, \Cref{partition} gives $s=\operatorname{poly}\left(\frac{d}{\varepsilon/(4d)}\right)=\operatorname{poly}\left(\frac{d^2}{\varepsilon}\right)$. 
		Thus, \(\log(2s+1)=O(\log d+\log(1/\varepsilon))\).
		For a sufficiently large absolute constant \(N\),
		the choice $d\ge\frac{N}{\varepsilon^2}$ implies $d\ge
		\frac{240}{\varepsilon}\log(2s+1)$. This completes the proof.
	\end{proof}

	\section{AI Disclosure}\label{app:ai-disclosure}

    We used OpenAI's ChatGPT 5.6 sol to assist with algorithms and proofs in \Cref{sec:induced-finder,sec:disconnected-tester}, exposition, and bibliographic checks, and Codex to assist with revisions. The main structural ideas were developed by the authors, who verified all AI-assisted content and take full responsibility for the paper.

\end{document}